\documentclass[11pt]{article}
\usepackage[utf8]{inputenc}
\usepackage{tcolorbox}
\usepackage{amsthm}
\usepackage{amsfonts}
\usepackage{amsmath}
\usepackage{amssymb}
\usepackage{colonequals}
\usepackage{epsfig}
\usepackage{url}
\newtheorem{theorem}{Theorem}
\newtheorem{corollary}[theorem]{Corollary}
\newtheorem{lemma}[theorem]{Lemma}
\newtheorem{observation}[theorem]{Observation}
\newtheorem{problem}[theorem]{Problem}
\newcommand{\GG}{{\cal G}}
\newcommand{\CC}{{\cal C}}

\newcommand{\II}{{\cal I}}
\newcommand{\RR}{{\cal R}}
\newcommand{\ZZ}{{\cal Z}}
\DeclareMathOperator{\tw}{\mathrm{tw}}

\DeclareMathOperator{\poly}{\mathrm{poly}}
\newcommand{\cin}[1]{\mathrm{in}_{#1}}
\newcommand{\prt}[1]{\mathrm{par}_{#1}}

\title{Approximation metatheorems for classes with bounded expansion\footnote{Supported by the ERC-CZ project LL2005 (Algorithms and complexity within and beyond bounded expansion) of the Ministry of Education of Czech Republic.}}
\author{Zden\v{e}k Dvo\v{r}\'ak\thanks{Computer Science Institute, Charles University, Prague, Czech Republic. E-mail: {\tt rakdver@iuuk.mff.cuni.cz}.}}
\date{}

\begin{document}
\maketitle

\begin{abstract}
We give a number of approximation metatheorems for monotone maximization problems expressible
in the first-order logic, in substantially more general settings than the previously known.
We obtain
\begin{itemize}
\item constant-factor approximation algorithm in any class of graphs with bounded expansion,
\item a QPTAS in any class with strongly sublinear separators, and
\item a PTAS in any fractionally treewidth-fragile class (which includes all common
classes with strongly sublinear separators.
\end{itemize}
Moreover, our tools also give an exact subexponential-time algorithm
in any class with strongly sublinear separators.
\end{abstract}

\section{Introduction}

We are interested in approximation algorithms for problems such as the \textsc{Maximum Independent Set}
and its variants (weighted, distance-$d$ independent for a fixed parameter $d$, \ldots),
\textsc{Maximum Induced Matching}, \textsc{Maximum $3$-Colorable Induced Subgraph}, and similar.
There are many strong non-approximation results that preclude the existence of constant-factor
approximation algorithms for these problems in general; for example, it is NP-hard to
approximate the independence number~\cite{colnonap} of an $n$-vertex graph up to the factor of $n^{1-\varepsilon}$ for
every $\varepsilon>0$.  Hence, we need to consider more restricted settings.

There is a close connection between approximability and the existence of algorithms parameterized by the solution size.
Indeed, the existence of an EPTAS (arbitrarily precise polynomial-time approximation algorithm such that the degree
of the polynomial bounding the complexity does not depend on the precision) directly implies fixed-parameter tractability,
and a constant-factor approximation often forms a starting point for proving fixed-parameter tractability.
A natural family of problems, namely those expressible in the first-order logic, is known to be fixed-parameter tractable
in a subgraph-closed class of graphs if and (under standard complexity-theoretic assumptions) only if
the class is nowhere-dense~\cite{grohe2014deciding}.  Moreover, in a slightly more restrictive setting of
classes with bounded expansion, the parameterized algorithms have linear time complexity~\cite{dvorak2013testing}.
We refer the reader not familiar with the concept of bounded expansion to Section~\ref{sec-bexp};
here, let us just mention that examples of graph classes with this property are
planar graphs and more generally all proper minor-closed classes, all graph
classes with bounded maximum degree, and even more generally, graph classes
closed under topological minors, as well as almost all Erd\H{o}s-R\'enyi random
graphs with bounded average degree.

Motivated by this connection, we explore the approximability of maximization problems expressible in the first-order logic
when restricted to classes with bounded expansion. As our first main result, we show that every monotone
maximization problems expressible in the first-order logic admits a constant-factor approximation algorithm
in every class of graphs of bounded expansion, even in the weighted setting.  We need a few definitions to formulate the precise statement.

Let $I$ be a finite index set and let $S$ be a set of vertices of a graph $G$.  An \emph{$I$-tuple of subsets} of $S$ is a system
$A_I=\{A_i:i\in I\}$, where $A_i\subseteq S$ for each $i$.  We say that the $I$-tuple \emph{covers} a vertex
$v\in V(G)$ if $v\in \bigcup_{i\in I} A_i$.  A \emph{property of $I$-tuples of subsets} in $G$ is
a set $\pi$ of $I$-tuples of subsets of $V(G)$, listing the $I$-tuples that \emph{satisfy the property $\pi$}.
As an example, suppose $I=\{1,2,3\}$ and $\pi$ consists exactly of the $I$-tuples $\{A_1,A_2,A_3\}$
such that $A_1$, $A_2$, and $A_3$ are disjoint independent sets in $G$; then an induced subgraph of $G$
is $3$-colorable if and only if its vertex set is covered by some $I$-tuple satisfying the property $\pi$.
For $I$-tuples $A_I$ and $A'_I$, we write $A'_I\subseteq A_I$ if
$A'_i$ is a subset of $A_i$ for each $i\in I$.  Similarly, we define $A_I\cup A'_I$, $A_I\cap X$ and $A_I\setminus X$ for a set $X\subseteq V(G)$
by applying the operation in each index separately.  We say that the property $\pi$ is \emph{monotone} if for all $I$-tuples
$A'_I\subseteq A_I$, if $A_I$ satisfies the property $\pi$, then so does $A'_I$.

Our goal will be to maximize the weight of an $I$-tuple satisfying the given property.  The weight of the $I$-tuple
is a sum of the weights of the covered vertices, where the weight of vertex is determined by its membership
in the elements of the $I$-tuple.  More precisely, for a vertex $v\in V(G)$, let $\chi_{A_I}(v)\in 2^I$
be the set of indices $i\in I$ such that $v\in A_i$.  A function $w:V(G)\times 2^I\to\mathbb{Z}$ is a \emph{weight assignment}
if $w(v,\emptyset)=0$ for each $v\in V(G)$ (with a few exceptions, we will only consider assignments of non-negative weights).
Let us define $w(A_I)=\sum_{v\in V(G)} w(v,\chi_{A_I}(v))$, and let
$\mathrm{MAX}(\pi,w)$ be the maximum of $w(A_I)$ over the systems $A_I$ satisfying the property $\pi$.

Let $X_I=\{X_i:i\in I\}$ be a system of unary predicate symbols (to be interpreted as subsets of vertices
of the input graph).  A \emph{first-order $I$-formula} is a formula $\varphi$ using quantification over vertices,
the predicates $X_i$ for $i\in I$, equality, and the standard logic conjunctions.  A \emph{first-order graph $I$-formula} can additionally
use a binary \emph{adjacency predicate} $E$.  A formula is a \emph{sentence} if it has no free variables.
For a graph $G$, an $I$-tuple $A_I$ of subsets of $V(G)$, and a first-order $I$-sentence $\varphi$,
we write $G,A_I\models\varphi$ if the sentence $\varphi$ holds
when the variables in its quantifiers take values from $V(G)$, the adjacency predicate is interpreted as the adjacency in $G$,
and for $i\in I$, $X_i$ is interpreted as the set $A_i$.  The \emph{property $\pi$ expressed by $\varphi$} consists
of all $I$-tuples $A_I$ such that $G,A_I\models\varphi$.  For example, the property ``$X_1$ is a distance-$2$ independent set''
is expressed by the first-order graph $\{1\}$-sentence
$$(\forall x,y)\, (X_1(x)\land X_1(y)\land x\neq y)\Rightarrow (\lnot E(x,y)\land \lnot(\exists z)\,E(x,z)\land E(y,z)).$$

\begin{theorem}\label{thm-bexpcap}
Let $I$ be a finite index set and let $\varphi$ be a first-order graph $I$-sentence
expressing a monotone property $\pi$.  For any graph class $\GG$ with bounded expansion, there exists
a constant $c\ge 1$ and a linear-time algorithm that, given
\begin{itemize}
\item a graph $G\in \GG$ and
\item a weight assignment $w:V(G)\times 2^I\to\mathbb{Z}_0^+$,
\end{itemize}
returns an $I$-tuple $A_I$ of subsets of $V(G)$ satisfying the property $\pi$ such that
$$w(A_I)\ge \tfrac{1}{c}\cdot\mathrm{MAX}(\pi,w).$$
\end{theorem}

Actually, the result applies to even more general class of properties, expressible by the fragment of monadic second-order
logic where we allow quantification only over the subsets of the vertices in the solution.  
For a finite index set $I$ disjoint from the integers, a
\emph{solution-restricted MSOL $I$-sentence with first-order graph core $\psi$} is a formula of form
$(Q_1X_1\subseteq \bigcup_{i\in I} X_i)\ldots (Q_nX_n\subseteq \bigcup_{i\in I} X_i)\;\psi$, where $Q_1$, \ldots, $Q_n$ are quantifiers,
$X_1$, \ldots, $X_n$ are unary predicate symbols
(interpreted as subsets of vertices of the input graph), and $\psi$ is a first-order graph $I\cup\{1,\ldots,n\}$-sentence.
For example, the property ``$G[X]$ is a union of cycles, and the distance in $G$ between the distinct cycles is at least three''
(or more natural properties such as ``$G[X]$ is planar'' or ``$G[X]$ is acyclic'' that however do not use the
fact that in the first-order core, we are allowed to quantify also over the vertices not in $X$) can be expressed
in this way.

\begin{theorem}\label{thm-msolcap}
Let $I$ be a finite index set and let $\varphi$ be a solution-restricted MSOL $I$-sentence with first-order graph core
expressing a monotone property $\pi$.  For any class $\GG$ with bounded expansion, there exists
a constant $c\ge 1$ and a linear-time algorithm that, given
\begin{itemize}
\item a graph $G\in \GG$ and
\item a weight assignment $w:V(G)\times 2^I\to\mathbb{Z}_0^+$,
\end{itemize}
returns an $I$-tuple $A_I$ of subsets of $V(G)$ satisfying the property $\pi$ such that
$$w(A_I)\ge \tfrac{1}{c}\cdot\mathrm{MAX}(\pi,w).$$
\end{theorem}

We are also interested in the graph classes for which every monotone
maximization problems expressible in the first-order logic admits a polynomial-time approximation scheme (PTAS),
i.e., an arbitrarily precise polynomial-time approximation algorithm.  Note that it is hard to approximate the maximum independent
set within the factor of $0.995$ in graphs of maximum degree at most three~\cite{berman1999some}, and thus we 
do not aim to obtain PTAS in all classes with bounded expansion.  The class of graphs of maximum degree three has exponential
expansion, motivating us to consider the classes with polynomial expansion.  Dvořák and Norin~\cite{dvorak2016strongly}
proved these are exactly the graph classes with strongly sublinear separators\footnote{For an $n$-vertex graph $G$, a set $X\subseteq V(G)$ is
a \emph{balanced separator} if each component of $G-X$ has at most $2n/3$ vertices.  Let $s(G)$ denote the minimum
size of a balanced separator in $G$, and for a class $\GG$ of graphs, let $s_\GG:\mathbb{Z}^+\to\mathbb{Z}_0^+$ be defined by
$s_\GG(n)=\max\{s(H):H\subseteq G\in\GG, |V(H)|\le n\}.$
The class $\GG$ has \emph{strongly sublinear separators} if $s_\GG(n)=O(n^{1-\beta})$ for some $\beta>0$.},
and the approximation questions have been intensively studied for various graph classes with this property
(such as planar graphs or more generally for proper minor-closed classes);
see Section~\ref{sec-previous} for an overview.  While we were not able to obtain a PTAS for all classes with strongly
sublinear separators, we were at least able to obtain a quasi-polynomial time approximation schemes
(but only for properties expressed by first-order graph sentences, rather than solution-restricted MSOL sentences
with first-order graph core).

\begin{theorem}\label{thm-qptas}
Let $I$ be a finite index set and let $\varphi$ be a first-order graph $I$-sentence
expressing a monotone property $\pi$.  For any class $\GG$ with strongly sublinear separators, there exists
a polynomial $p$ and an algorithm that, given
\begin{itemize}
\item a graph $G\in \GG$,
\item a weight assignment $w:V(G)\times 2^I\to\mathbb{Z}_0^+$, and
\item a positive integer $o$,
\end{itemize}
returns in time $\exp(p(o\cdot \log |V(G)|))$ an $I$-tuple $A_I$ of subsets of $V(G)$ satisfying the property $\pi$ such that
$$w(A_I)\ge \bigl(1-\tfrac{1}{o}\bigr)\cdot\mathrm{MAX}(\pi,w).$$
\end{theorem}

Interestingly, the ideas used to prove Theorem~\ref{thm-qptas} also lead to exact subexponential-time algorithms.
\begin{theorem}\label{thm-exact}
Let $I$ be a finite index set and let $\varphi$ be a first-order graph $I$-sentence
expressing a property $\pi$.  Let $\GG$ be a class of graphs such that $s_\GG(n)=O(n^{1-\beta})$ for some positive $\beta<1$.
There exists an algorithm that, given
\begin{itemize}
\item a graph $G\in \GG$ and
\item a weight assignment $w:V(G)\times 2^I\to\mathbb{Z}$ (of not necessarily non-negative weights),
\end{itemize}
returns in time $\exp(O(|V(G)|^{1-\beta}\log^{1/2} |V(G)|))$ an $I$-tuple $A_I$ of subsets of $V(G)$ satisfying the property $\pi$ such that
$w(A_I)=\mathrm{MAX}(\pi,w)$.
\end{theorem}

We can obtain PTASes under a slightly stronger assumption on the considered class of graphs,
efficient \emph{fractional treewidth-fragility}.  For a positive integer $s$ and a positive real number $\delta\le 1$,
a multiset $\ZZ$ of subsets of vertices of a graph $G$ is an \emph{$(s,\delta)$-generic cover} of $G$
if for every set $S\subseteq V(G)$ of size at most $s$, we have $S\subseteq Z$ for at least $\delta|\ZZ|$ sets $Z\in\ZZ$.
The \emph{treewidth} of the cover is the maximum of $\tw(G[Z])$ over all $Z\in\mathbb{Z}$.
We say that a class of graphs $\GG$ is \emph{fractionally treewidth-fragile}
if for some function $f:\mathbb{Z}^+\to\mathbb{Z}^+$, the following claim holds:
for every $G\in\GG$ and positive integers $s$ and $o$, there exists an $(s,1-1/o)$-generic cover of $G$
of treewidth at most $f(os)$.  The class is \emph{efficiently fractionally treewidth-fragile}
if such a cover can be found in time polynomial in $|V(G)|$, and in particular, $\ZZ$ has polynomial size\footnote{Note this
is a somewhat non-standard formulation of fractional treewidth-fragility.  In the usual definition~\cite{twd,onsubsep},
one requires the existence of a system of sets whose deletion results in a graph of treewidth at most $f(o)$ and such that each vertex
belongs to at most $1/o$ fraction of the sets, i.e., the complements of the sets of the system form a $(1,1-1/o)$-generic
cover of treewidth at most $f(o)$.   To match this with our definition, it suffices to observe that a $\bigl(1,1-\tfrac{1}{os}\bigr)$-generic
cover is also $(s,1-1/o)$-generic.}.

\begin{theorem}\label{thm-msolptas}
Let $I$ be a finite index set and let $\varphi$ be a solution-restricted MSOL $I$-sentence with first-order graph core
expressing a monotone property $\pi$.  For any class $\GG$ that is efficiently fractionally treewidth-fragile, there exists
a function $f$, a polynomial $p$, and an algorithm that, given
\begin{itemize}
\item a graph $G\in \GG$,
\item a weight assignment $w:V(G)\to\mathbb{Z}_0^+$, and
\item a positive integer $o$,
\end{itemize}
returns in time $f(o)p(|V(G)|)$ an $I$-tuple $A_I$ of subsets of $V(G)$ satisfying the property $\pi$ such that
$$w(A_I)\ge \bigl(1-\tfrac{1}{o}\bigr)\cdot\mathrm{MAX}(\pi,w).$$
\end{theorem}

Efficiently fractionally treewidth-fragile classes include many of the known graph classes with strongly sublinear separators, in particular
\begin{itemize}
\item all hereditary classes with sublinear separators and bounded maximum degree~\cite{twd},
\item all proper minor-closed classes, as an easy consequence of the result of DeVos et al.~\cite{devospart}
(or~\cite{bakergame} without using the Robertson-Seymour structure theorem), and
\item many geometric graph classes, such as the intersection graphs of convex sets with bounded aspect ratio in a fixed Euclidean space
that have bounded clique number (as can be seen using the idea of~\cite{erlebach2005polynomial}).
\end{itemize}
Indeed, it is possible (and I have conjectured) that all classes with sublinear separators are fractionally treewidth-fragile.

Let us finish the introduction by giving two natural open questions.
Our results only apply to maximization problems.  More precisely, the technique we use can be applied to minimization problems
(with \emph{monotone} meaning the supersets of valid solutions are also valid solutions)
as well, but the resulting algorithm have error bounded by a fraction of the total weight of all vertices,
rather than a fraction of the optimal solution weight.
\begin{problem}
Do monotone minimization problems expressible in the first order logic admit constant factor approximation
in all classes with bounded expansion?  And PTASes in all efficiently fractionally treewidth-fragile
graph classes?
\end{problem}
As a simplest example, we do not know whether there exists a PTAS for weighted vertex cover in fractionally treewidth-fragile graph classes.

Secondly, many results for classes of graphs with bounded expansion extend to nowhere-dense graph classes,
up to replacement of some constants by terms of order $n^{o(1)}$.  Our approach does not apply to this setting,
since it is based on a quantifier elimination result specific to graph classes with bounded expansion.
\begin{problem}
Do monotone maximization problems expressible in the first order logic admit an $O(n^{o(1)})$-factor approximation
for $n$-vertex graphs from nowhere-dense classes?
\end{problem}

\subsection{Proof outline}

The proofs of all our results are based on three ingredients:

\textbf{(1) A strong locality result for first-order properties in graphs from
classes with bounded expansion}, proved using a modification of the quantifier elimination procedure of~\cite{dvorak2013testing}.
To state the result, we need a few more definitions.  A \emph{simple signature} $\sigma$ is a set of unary predicate and function symbols.
For a finite index set $I$ and a system $X_I$ of unary predicate symbols disjoint from $\sigma$, a
\emph{first-order graph $(I,\sigma)$-formula} is a formula $\varphi$ using all the ingredients from
the definition of a first-order graph $I$-formula and additionally the predicates and unary functions from $\sigma$.
For a graph $G$, a unary function $f:V(G)\to V(G)$ is \emph{guarded by $G$} if for each $v\in V(G)$, either $f(v)=v$ or
$v$ is adjacent to $f(v)$ in $G$.  A \emph{$G$-interpretation} $\II$ of $\sigma$ assigns to each unary predicate symbol $P$ a subset $P_\II$ of vertices of $G$
and to each unary function symbol $f$ a unary function $f_\II$ guarded by $G$.  
For a positive integer $s$ and a graph $G$, a function $h:V(G)\to 2^{V(G)}$ is an \emph{$s$-shroud} if for each $v\in V(G)$,
$|h(v)|\le s$ and $v\in h(v)$.  An \emph{$h$-center} of a set $Y\subseteq V(G)$ is the set $\{v\in Y:h(v)\subseteq Y\}$.

\begin{theorem}\label{thm-local}
Let $I$ be a finite index set, let $\varphi$ be a first-order graph $I$-sentence,
and let $\GG$ be a class of graphs with bounded expansion.
There exists a constant $s$, a simple signature $\sigma$ disjoint from all symbols appearing in $\varphi$,
and a first-order graph $(I,\sigma)$-sentence $\varphi'$ such that the following claim holds.

Given a graph $G\in \GG$, we can in linear time find an $s$-shroud $h$ with the following property: For any $Y\subseteq V(G)$,
we can in linear time find a $G[Y]$-interpretation $\II_Y$ of $\sigma$ for which every $I$-tuple $A_I$ of subsets of
the $h$-center of $Y$ satisfies
$$G,A_I\models \varphi\text{ if and only if }G[Y],\II_Y,A_I\models \varphi'.$$
\end{theorem}

That is, for the $I$-tuples of subsets of the $h$-center of $Y$, we can evaluate whether they have the property $\pi$
(in the whole graph $G$) just by looking at the induced subgraph $G[Y]$ enhanced by $\II_Y$.  Note that
Theorem~\ref{thm-local} straightforwardly extends to solution-restricted MSOL $I$-sentences 
$(Q_1X_1\subseteq \bigcup_{i\in I} X_i)\ldots (Q_nX_n\subseteq \bigcup_{i\in I} X_i)\;\psi$
with first-order graph core, since if $X_I$ is interpreted as an $I$-tuple $A_I$ of subsets of the $h$-center of $Y$,
then $X_1, \ldots, X_n$ also correspond to subsets of the $h$-center of $Y$.

\begin{corollary}\label{cor-local}
Let $I$ be a finite index set, let $\varphi$ be a solution-restricted MSOL $I$-sentence with first-order graph core,
and let $\GG$ be a class of graphs with bounded expansion.
There exists a constant $s$, a simple signature $\sigma$ disjoint from all symbols appearing in $\varphi$,
and a solution-restricted MSOL $(I,\sigma)$-sentence $\varphi'$ with first-order graph core such that the following claim holds.

Given a graph $G\in \GG$, we can in linear time find an $s$-shroud $h$ with the following property: For any $Y\subseteq V(G)$,
we can in linear time find a $G[Y]$-interpretation $\II_Y$ of $\sigma$ for which every $I$-tuple $A_I$ of subsets of
the $h$-center of $Y$ satisfies
$$G,A_I\models \varphi\text{ if and only if }G[Y],\II_Y,A_I\models \varphi'.$$
\end{corollary}

\textbf{(2) The existence of sufficiently generic covers.}  For efficiently fractionally
treewidth-fragile classes, we have them by definition.  For classes with bounded expansion, we use covers
obtained from \emph{low-treedepth colorings}.  A \emph{rooted forest} $F$ is an acyclic graph with a specified \emph{root}
vertex in each component.  The \emph{depth} of $F$ is the number of vertices
on the longest path from a root to a leaf.  If the path in $F$ from a root to a vertex $v$ contains a vertex $u$,
we say that $u$ is an \emph{ancestor} of $v$ and $v$ is a \emph{descendant} of $u$.  The \emph{closure} of $F$ is the graph with
the vertex set $V(F)$ where each vertex is adjacent exactly to its ancestors and descendants in $F$.  The \emph{treedepth}
of a graph $H$ is the minimum $d$ such that $H$ is a subgraph of the closure of a rooted forest of depth $d$.
A graph of treedepth $d$ is known to have treewidth (in fact, even pathwidth) smaller than $d$~\cite{nesbook}.
For a positive integer $s$, a \emph{treedepth-$s$ coloring} of a graph $G$ is a coloring such that the union of
every $s$ color classes induces a subgraph of treewidth at most $s$.  Ne\v{s}et\v{r}il and Ossona de Mendez~\cite{grad2}
proved the following claim.
\begin{theorem}[Ne\v{s}et\v{r}il and Ossona de Mendez~\cite{grad2}]\label{thm-lowtw}
For every class $\GG$ of graphs with bounded expansion and every positive integer $s$, there exists an integer $a$ and a linear-time
algorithm that given a graph $G\in \GG$ returns a treedepth-$s$ coloring of $G$ using at most $a$ colors.
\end{theorem}
By considering the cover consisting of all $\binom{a}{s}$ unions of $s$-tuples of color classes, we obtain the following claim.
\begin{corollary}\label{cor-lowtd}
For every class $\GG$ of graphs with bounded expansion and every positive integer $s$, there exists a positive integer $c$ and a linear-time
that given a graph $G\in \GG$ returns an $(s,1/c)$-generic cover of size $c$ and treedepth at most $s$.
\end{corollary}
For classes with strongly sublinear separators, in~\cite{onsubsep} I proved they are ``almost'' fractionally treewidth-fragile,
in the following sense (again, with a somewhat different notation, see the footnote at the definition of fractional treewidth-fragility).
\begin{theorem}[Dvořák~\cite{onsubsep}]\label{thm-pllcov}
For every class $\GG$ with strongly sublinear separators, there exists a polynomial $f:\mathbb{Z}^+\to\mathbb{Z}^+$
and a polynomial-time algorithm that, for every $G\in\GG$ and positive integers $s$ and $o$,
returns an $(s,1-1/o)$-generic cover of $G$ of treewidth at most $f(os\log |V(G)|)$.
Moreover, the algorithm also returns the corresponding tree decomposition for each element of the cover.
\end{theorem}

\textbf{(3) The means to solve the problem on graphs of bounded treewidth.}  For Theorems~\ref{thm-msolcap} and \ref{thm-msolptas},
we use the well-known result of Courcelle~\cite{courcelle}, in the following optimization version (note that we do not need
to be given a tree decomposition, as for graphs of bounded treewidth, we can find an optimal tree decomposition in linear time~\cite{bodlaender1996linear}).

\begin{theorem}\label{thm-optmsol}
Let $I$ be a finite index set and let $\varphi$ be a MSOL graph formula with free variables $X_I$,
expressing a property $\pi$.  For any positive integer $t$, there exists a linear-time algorithm that, given
\begin{itemize}
\item a graph $G$ of treewidth at most $t$,
\item a set $X\subseteq V(G)$, and
\item a weight assignment $w:V(G)\times 2^I\to\mathbb{Z}$,
\end{itemize}
returns an $I$-tuple $A_I$ of subsets of $X$ satisfying the property $\pi$ such that
$$w(A_I)=\mathrm{MAX}(\pi,w).$$
\end{theorem}

This is not sufficient for the proof of Theorem~\ref{thm-qptas}, where we work with covers of polylogarithmic treewidth,
and thus we need a better control over the dependence of the time complexity on the treewidth.  We use the following result
proved using the locality property underlying Theorem~\ref{thm-local}; we believe this result to be of independent interest.
\begin{theorem}\label{thm-optfo}
Let $I$ be a finite index set and let $\varphi$ be a first-order graph $I$-sentence
expressing a property $\pi$.  For any class $\GG$ with bounded expansion, there exists
a constant $c>0$ and an algorithm that, given
\begin{itemize}
\item a graph $G\in \GG$,
\item a tree decomposition $\tau$ of $G$ with at most $|V(G)|$ nodes,
\item a set $X\subseteq V(G)$, and
\item a weight assignment $w:V(G)\times 2^I\to\mathbb{Z}_0$,
\end{itemize}
returns in time $O(\exp(ct)|V(G)|)$, where $t$ is the width of the decomposition $\tau$, an $I$-tuple $A_I$ of subsets of $X$ satisfying the property $\pi$ such that
$$w(A_I)=\mathrm{MAX}(\pi,w).$$
\end{theorem}
Let us remark that the reason we are not able to generalize Theorem~\ref{thm-qptas} to solution-restricted MSOL sentences
with first-order graph core is that we cannot prove the analogue of Theorem~\ref{thm-optfo} in that setting.

We are now ready to prove our results.
\begin{proof}[Proof of Theorems~\ref{thm-msolcap}, \ref{thm-qptas} and \ref{thm-msolptas}]
Note that $\GG$ has bounded expansion:
\begin{itemize}
\item In the situation of Theorem~\ref{thm-msolcap}, this is an assumption.
\item In the situation of Theorem~\ref{thm-qptas}, this is the case since classes with strongly sublinear separators have polynomial expansion~\cite{dvorak2016strongly}.
\item In the situation of Theorem~\ref{thm-msolptas}, this is the case since every fractionally treewidth-fragile class has bounded expansion~\cite{twd}.
\end{itemize}
Let $s$, $\sigma$, and $\varphi'$ be obtained by applying Theorem~\ref{thm-local} to $\varphi$ and $\GG$.
Now, for the input graph $G$ and the weight assignment $w$ (and the precision $o$ in the case of Theorems~\ref{thm-qptas} and \ref{thm-msolptas}),
we apply the following algorithm:
\begin{itemize}
\item Let $\ZZ$ be an $(s,\delta)$-generic cover of treewidth at most $t$, where
\begin{itemize}
\item in the situation of Theorem~\ref{thm-msolcap}, $\ZZ$ is obtained using Corollary~\ref{cor-local},
$\delta=1/c$, and $t=s$;
\item in the situation of Theorem~\ref{thm-msolptas}, $\ZZ$ is obtained using Theorem~\ref{thm-pllcov}
for the given $o$ and $s$, $\delta=1-1/o$, and $t=f(os\log |V(G)|)$; and
\item in the situation of Theorem~\ref{thm-msolptas}, $\ZZ$ is obtained using the definition of
efficient fractional treewidth-fragility for the given $o$ and $s$, $\delta=1-1/o$, and $t=f(os)$.
\end{itemize}
\item Let $h$ be the $s$-shroud obtained using Theorem~\ref{thm-local}.
\item For each $Y\in \ZZ$:
\begin{itemize}
\item Let $\II_Y$ be the $G[Y]$-interpretation of $\sigma$ obtained using Theorem~\ref{thm-local}.
\item Let $Y'$ be the $h$-center of $Y$.
\item Using Theorem~\ref{thm-optmsol} or \ref{thm-optfo} in the bounded-treewidth graph $G[Y]$,
find an $I$-tuple $A_I^Y$ of subsets of $Y'$ satisfying $G[Y],\II_Y,A^Y_I\models \varphi'$ such that $w(A^Y_I)$ is maximum
possible.
\end{itemize}
\item Return the $I$-tuple $A_I^Y$ such that $w(A^Y_I)$ is maximum among all $Y\in\ZZ$.
\end{itemize}
Note that Theorem~\ref{thm-local} implies that $G,A^Y_I\models \varphi$, and thus it suffices to bound the
approximation ratio of the algorithm.

Let $A_I$ be an $I$-tuple of subsets of vertices of $G$ satisfying the property $\pi$ such that $w(A_I)$ is maximum.
Choose $Y\in \ZZ$ uniformly at random.  Since $h$ is an $s$-shroud and $\ZZ$ is an $(s,\delta)$-cover,
for each $v\in V(G)$, the probability that $h(v)\subseteq Y$ is at least $\delta$.
Hence, letting $A'_I=A_I\cap Y'$ (where $Y'$ is the $h$-center of $Y$)
the expected value of $w(A'_I)$ is at least $\delta\cdot w(A_I)$.
Since $\pi$ is monotone, we have $G,A'_I\models \varphi$, and by Theorem~\ref{thm-local},
$G,\II_Y,A'_I\models\varphi'$.  This implies $w(A^Y_I)\ge w(A'_I)$, and
thus the expected value of $w(A^Y_I)$ is also at least $\delta\cdot w(A_I)$.
Since we return the maximum over all elements of $\ZZ$, this implies the weight
of the returned set is at least $\delta\cdot w(A_I)$.
\end{proof}

\begin{proof}[Proof of Theorem~\ref{thm-exact}]
Note that since $\GG$ has strongly sublinear separators, it has bounded (in fact, polynomial) expansion~\cite{dvorak2016strongly}.

Using the algorithm of~\cite{feige}, we can for any $n$-vertex subgraph of $G$ in polynomial time find a balanced separator of size
$O(n^{1-\beta}\log^{1/2} n)$.  Using this algorithm, let us construct a tree decomposition $\tau$ of $G$ as follows:
\begin{itemize}
\item Find a balanced separator $S$ in $G$.
\item Recursively find a tree decomposition $\tau_C$ of each component $C$ of $G-S$.
\item Add a new root vertex adjacent to the root of $\tau_C$ for each each component $C$ of $G-S$,
and add $S$ to all bags including the root one.
\end{itemize}
This tree decomposition has width $t=O(|V(G)|^{1-\beta}\log^{1/2} |V(G)|)$.  We then apply the algorithm from Theorem~\ref{thm-optfo}
for this tree decomposition and $X=V(G)$.
\end{proof}

Theorems~\ref{thm-local} and \ref{thm-optfo} are proved in Section~\ref{sec-elim}.

\subsection{Bounded expansion}\label{sec-bexp}

The theory of bounded expansion and nowhere-density was developed chiefly by Nešetřil and Ossona de Mendez
in a series of papers~\cite{grad1,grad2,npom-nd1} to capture the notion of graph sparsity with respect to
the expressive power of the first-order logic.

For a non-negative integer $r$, an \emph{$r$-shallow minor} of a graph $G$
is any graph obtained from a subgraph of $G$ by contracting pairwise vertex-disjoint subgraphs of radius at most $r$.
A class of graphs $\GG$ has \emph{expansion bounded by a function $f:\mathbb{Z}_0^+\to\mathbb{Z}_0^+$}
if for every $r\ge 0$, every $r$-shallow minor of a graph belonging to $\GG$ has average degree at most $f(r)$.
We say that a class has \emph{bounded expansion} if it has expansion bounded by some function,
and \emph{polynomial expansion} if it has expansion bounded by a polynomial.

Many natural graph classes have bounded expansion, thus making it possible to treat them uniformly within this
framework.  For example, Dvořák and Norin~\cite{dvorak2016strongly} proved that a class of graphs has polynomial expansion
if and only if it has strongly sublinear separators.  This includes
\begin{itemize}
\item planar graphs~\cite{lt79}, and more generally all proper minor-closed classes~\cite{alon1990separator};
\item graphs drawn in the plane (or on a fixed surface) with a bounded number of crossings on each edge~\cite{osmenwood}; and
\item many geometric graph classes, such as the intersection graphs of convex sets with bounded aspect ratio in a fixed Euclidean space
that have bounded clique number, or nearest-neighbor graphs of point sets in a fixed Euclidean space~\cite{miller1997separators}.
\end{itemize}
Classes with bounded (but superpolynomial) expansion include
\begin{itemize}
\item Graph classes with bounded maximum degree, and more generally all graph classes
closed under topological minors~\cite{grad1};
\item graphs with bounded stack or queue number~\cite{osmenwood}; and,
\item almost all Erd\H{o}s-R\'enyi random graphs with linear number of edges~\cite{osmenwood}.
\end{itemize}

For a more in-depth introduction to the topic, the reader is referred to the book of Ne\v{s}et\v{r}il and Ossona de Mendez~\cite{nesbook}.
The classes with bounded expansion have found many applications in the design of parameterized
algorithms~\cite{dvorak2013testing,domker,amiri2017distributed,kreutzer2018polynomial,kerdom}.
Their applications in he context of approximation algorithms are discussed in the following section.

\subsection{Related work}\label{sec-previous}

Distance versions of both minimum dominating set and the maximum independent set are known to admit
constant-factor approximation algorithms in classes with bounded expansion~\cite{apxdomin,apxlp}.
A constant-factor approximation algorithm for weighted and distance version of the minimum dominating set 
also follows from~\cite{chan2012weighted} combined with the bounds on the neighborhood complexity
in classes with bounded expansion~\cite{neicom}.

There are many knonw techniques to obtain approximation schemes for specific
classes of graphs with strongly sublinear separators (polynomial expansion).
We illustrate the power of these techniques on variants of the maximum \textsc{independent set} problem:
The \textsc{distance-$r$ independent set} problem (parameterized by a fixed positive integer $r$),
where we require the distance between distinct vertices of the chosen set to be greater than $r$,
and the \textsc{weighted} version of the problem, where the input contains an assignment of weights to vertices
and we maximize the sum of weights of vertices in the set rather than the size of the set;
see Table~\ref{tab-compare} for a summary.
\begin{itemize}
\item Let us start with the techniques that apply to all classes with strongly sublinear separators.
Lipton and Tarjan~\cite{lt80} observed that for each $\varepsilon>0$,  one can split the input graph $G$
by iteratively deleting sublinear separators into components of size $\poly(1/\varepsilon)$,
where the resulting set $R$ of removed vertices has size at most $\varepsilon|V(G)|$.
One can then solve the problem in each component separately by brute force and obtain an approximation with the additive error
$\varepsilon|V(G)|$.  In addition to only giving an additive approximation bound, this technique is limited to
the problems for which a global solution can be obtained from the partial solutions in $G-R$; e.g., it does not apply to
the \textsc{distance-$2$ independent set} problem.  It also does not apply in the \textsc{weighted} setting.

Har-Peled and Quanrud~\cite{har2015approximation} proved that in any hereditary class with sublinear separators, a simple
local search approach (incrementally improving an initial solution by changes of bounded size) gives PTAS for a number
of natural optimization problems, including the \textsc{$r$-independent set} problem for any fixed $r\ge 1$ (this
is not explicitly stated in their paper, but it is easy to work out the argument).  On the other hand,
it is not clear which problems are amenable to this approach, and it fails even for some very simple problems
(e.g., finding the maximum monochromatic set in an edgeless graphs with vertices colored red and blue).
The technique does not apply in the weighted setting.
\item The property of \emph{fractional treewidth-fragility}, which as noted
above is satisfied by many classes with strongly sublinear separators, was developed as a way to extend
Baker's technique (discussed below) to more general graph classes.  While a direct application (solving the problem
separately in each subgraph induced by the cover) fails for the distance versions of the problems,
we have overcame this restriction in a joint work with Lahiri~\cite{dvorlah}.  The approach presented in this
paper can be seen as a substantial generalization in terms of the algorithmic problems to which
it applies (Dvořák and Lahiri~\cite{dvorlah} only consider problems expressible in terms of distances between
the solution vertices).  The approach works for weighted problems.  However, as a major restriction,
it generally only applies to maximization problems.
\item In~\cite{thin}, I made a rather technical attempt to improve upon the fractional treewidth-fragility,
by introducing a more powerful notion of \emph{thin systems of overlays}.  All hereditary classes with sublinear separators
and bounded maximum degree have this property, and so do all proper minor-closed classes.  Thin systems of overlays make
it possible to design PTAS for the \textsc{$r$-independent set} problem for any fixed $r\ge 1$, as well as for other problems defined in terms of distances
between the vertices in the solution (including minimization problems such as the distance version of the minimum dominating set).
However, the notion is not suitable for the problems where more complex relationships need to be considered.
\item Baker~\cite{baker1994approximation} designed a very powerful technique for planar graphs based on finding a partition
of the graph into layers (where the edges are allowed only within the layers and between consecutive layers)
such that the union of a bounded number of these layers induces a subgraph of bounded treewidth.
This is a substantially more restrictive condition than fractional treewidth-fragility (not even all proper minor-closed classes
have this property).  In a trade-off, the range of problems for which it applies is much wider.  In particular, it can deal with
all (maximization or minimization) problems expressible in monotone first-order logic~\cite{dawar2006approximation}, which includes all the discussed variants
of the \textsc{independent set} problem.  A modified version of this technique
(where the layering step is iterated and combined with removal of a bounded number of vertices) also can be used in less restricted
settings~\cite{bakergame}, including for example all proper minor-closed classes (but not all classes with sublinear separators).
\item The arguments based on bidimensionality~\cite{demaine2005bidimensionality} are rather powerful, but limited in scope to (subclasses of)
the proper minor-closed classes. With regards to the PTAS design, they essentially build on the Baker's technique framework.
\end{itemize}

\begin{table}
\begin{tabular}{|p{2cm}|p{3cm}||c|c|c|c|}
\hline
Technique&Applies to&\textsc{indep. set}&\textsc{$r$-indep. s.}&\textsc{weighted i.s.}&monot. FO\\
\hline
\hline
Iterated separators&all classes with sublinear separators&\checkmark&&&\\
\hline
Local search&all classes with sublinear separators&\checkmark&\checkmark&&\\
\hline
Fractional treewidth-fragility&bounded max. degree, proper minor closed, \ldots; maybe all?&\checkmark&\checkmark&\checkmark&\checkmark (this work)\\
\hline
Thin systems of overlays&bounded max. degree, proper minor closed, \ldots&\checkmark&\checkmark&\checkmark&\\
\hline
Baker's technique&proper minor closed, some geometric settings&\checkmark&\checkmark&\checkmark&\checkmark\\
\hline
\end{tabular}
\caption{PTAS design techniques in hereditary classes with sublinear separators.}\label{tab-compare}
\end{table}

\section{The quantifier elimination result and its applications}\label{sec-elim}

In~\cite{twd}, it has been shown that any first-order formula on a graph from a class with bounded expansion
can be transformed into an equivalent quantifier-free first-order formula on a graph from a (different) class with bounded expansion,
by introducing unary functions guarded by the resulting graph and new unary predicates, both computable in linear time.
Essentially the same procedure can be applied to first-order formulas describing set properties; however, for our application,
we need to be more explicit in terms of how the newly introduced functions and predicates are defined and in particular, how
they depend on the set system whose weight we are maximizing.

To this end, in addition to unary functions and predicates, we introduce \emph{counters}.  Semantically, a counter $\gamma$ assigns a non-negative
integer $\gamma(v)$ to each vertex $v$.  For each counter symbol $\gamma$, our formulas can use expressions of form $\gamma(x)\ge m$,
where $x$ is a term and $m$ is a positive integer, with the natural interpretation.
The symbols for counters are linearly ordered, and we say that a formula is \emph{$\gamma$-dominated} if all counter symbols appearing in the formula are strictly smaller than $\gamma$.
For a variable $x$, we say that a formula $\theta$ is \emph{$x$-local}
if $\theta$ is quantifier-free, does not use unary functions, and $x$ is the only variable appearing in $\theta$.

The counters are used to keep track of the numbers of vertices that satisfy a prescribed property.
Formally, the values of each counter $\gamma$ are determined by an associated \emph{trigger} $(f,\theta)$, where $f$ is a unary function symbol
and $\theta$ is an $x$-local $\gamma$-dominated first-order formula.  For each vertex $v$, the value $\gamma(v)$ is equal to the number
of vertices $u\in V(G)\setminus\{v\}$ such that $f(u)=v$ and $\theta(u)$ holds.

A \emph{global formula} is a formula that can additionally use elementary formulas of form $\#\theta\ge m$, where $\theta$
is an $x$-local formula and $m$ is a positive integer; this elementary formula is true if there are at least $m$
vertices $v\in V(G)$ such that $\theta(v)$ holds.

Let $I$ be a finite index set.  
A \emph{counter $I$-signature} $\sigma$ is a set of unary predicate and function symbols and linearly ordered counter symbols,
together with the triggers associated with these counter symbols, where the triggers are allowed to refer to the unary predicates from $X_I$ and $\sigma$.
For a graph $G$, a \emph{$G$-interpretation} of $\sigma$ is a $G$-interpretation
of the simple signature consisting of the unary predicate and function symbols from $\sigma$.
We can now state the quantifier elimination result.  

\begin{theorem}\label{thm-elim}
Let $I$ be a finite index set, let $\varphi$ be a first-order graph $I$-sentence,
and let $\GG$ be a class of graphs with bounded expansion.  There exists a counter $I$-signature $\sigma$
and a global quantifier-free first-order $(I,\sigma)$-sentence $\varphi'$ such that the following claim holds.
Given a graph $G\in \GG$, we can in linear time compute a $G$-interpretation $\II$ of $\sigma$ such that
$$G,A_I\models \varphi\text{ if and only if }G,\II,A_I\models \varphi'$$
for every $I$-tuple $A_I$ of subsets of $V(G)$.
\end{theorem}

Let us remark that since $\varphi'$ is a quantifier-free sentence, it is a Boolean combination of
formulas of form $\#\theta\ge m$, where $\theta$ is an $x$-local formula, and in particular $\varphi'$ does not
refer to any function symbols.

Importantly, the interpretation $\II$ of the predicate and function symbols from $\sigma$ is independent on $A_I$;
the choice of $A_I$ only affects the values of counters.  Let $\ell(\sigma)$ denote the number of counter symbols in $\sigma$,
and for $i=1,\ldots,\ell(\sigma)$, let $f_{i,\II}$ be the interpretation of the unary function from the trigger of the
$i$-th counter symbol in $\sigma$, in their fixed linear ordering.  For $v\in V(G)$, let us define
$h_{0,\II}(v)=\{v\}$ and for $i=1,\ldots,\ell(\sigma)$, let
$h_{i,\II}(v)=h_{i-1,\II}(v)\cup \{f_{i,\II}(u):u\in h_{i-1,\II}(v)\}$.  We let $h_\II=h_{\ell(\sigma),\II}$;
note that $h_\II$ is a $2^{\ell(\sigma)}$-shroud.

For a counter $\gamma$ of a counter $I$-signature $\sigma$, a $G$-interpretation $\II$, and an $I$-tuple $A_I$ of subsets of $V(G)$,
let $\gamma(\II,A_I,v)$ denote the value of the counter $\gamma$ when the symbols of $\sigma$ are interpreted
according to $\II$ and the predicates $X_I$ are interpreted as $A_I$.  The applications of Theorem~\ref{thm-elim} use the
simple fact that the membership of a vertex $v$ in the sets of an $I$-tuple $A_I$ only affects the values of the
counters in $h_\II(v)$.

\begin{lemma}\label{lemma-shroud}
Let $I$ be a finite index set, let $\sigma$ be a counter $I$-signature, let $G$ be a graph and let $\II$
be a $G$-interpretation of $\sigma$.  For any $X\subseteq V(G)$ and $I$-tuples $A_I$ and $A'_I$ of subsets of vertices
of $G$, if $A_I\setminus X=A'_I\setminus X$, then $\gamma(\II,A_I,v)=\gamma(\II,A'_I,v)$
for every $v\in V(G)\setminus h_\II(X)$.
\end{lemma}
\begin{proof}
For $i=1,\ldots,\ell(\sigma)$, let $(\theta_i,f_i)$ be the trigger of the
$i$-th counter symbol $\gamma_i$ in $\sigma$, in their fixed linear ordering, and let
$f_{i,\II}$ be the interpretation of $f_i$ in $\II$.  By induction on $i$, we show that
$\gamma_i(\II,A_I,v)=\gamma_i(\II,A'_I,v)$ for each $v\in V(G)\setminus h_{i,\II}(X)$.
Recall that $\gamma_i(\II,A_I,v)$ is the number of vertices $u\in V(G)\setminus\{v\}$
such that $f_{i,\II}(u)=v$ and $G,\II,A_I\models \theta_i(u)$.  Note that
$u\in V(G)\setminus h_{i-1,\II}(X)\subseteq V(G)\setminus h_{j,\II}(X)\subseteq V(G)\setminus X$ for each $j\le i-1$,
and thus $\gamma_j(\II,A_I,u)=\gamma_j(\II,A'_I,v)$ for each such $j$ by the induction hypothesis.
Since $\theta_i$ is $x$-local, the value of $\theta_i(u)$ only depends on these counters,
$\II$, and the interpretation of $X_I$ on $u\not\in X$,
and thus $G,\II,A_I\models \theta_i(u)$ if and only if $G,\II,A'_I\models \theta_i(u)$.
Consequently, $\gamma_i(\II,A_I,v)=\gamma_i(\II,A'_I,v)$, as required.
\end{proof}

This is useful in combination with the fact that vertices with fixed values of the counters can be
deleted.  Let $I$ be a finite index set, let $\sigma$ be a counter $I$-signature, and let $\varphi$ be a global quantifier-free first-order
$(I,\sigma)$-sentence.
For a graph $G$ and $Y\subseteq V(G)$, a \emph{$(G,Y,\sigma,\varphi)$-census} is a function $n$ that, letting $M$ be the largest integer appearing in the formula $\varphi$ and the triggers of $\sigma$,
assigns an element of $\{0,\ldots,M\}$ to
\begin{itemize}
\item each $x$-local formula $\theta$ appearing in an elementary formula $\#\theta\ge m$ in $\varphi$, and
\item each pair $(\gamma,v)$, where $\gamma$ is a counter symbol of $\sigma$ and $v\in Y$ is a vertex with at least one neighbor in $V(G)\setminus Y$.
\end{itemize}
Given a $G$-interpretation $\II$ of $\sigma$ and an $I$-tuple $A_I$ of subsets of vertices of $G$, we say that $n$ is a \emph{$(G,Y,\II,\sigma,\varphi)$-shadow of $A_I$}
if
\begin{itemize}
\item for each $x$-local formula $\theta$ appearing in an elementary formula $\#\theta\ge m$ in $\varphi$,
$n(\theta)$ is the minimum of $M$ and the number of vertices $u\in V(G)\setminus Y$ such that
$G,\II,A_I\models \theta(u)$; and,
\item for each counter $\gamma$ with trigger $(\theta,f)$ and each vertex $v\in Y$ with at least
one neighbor in $V(G)\setminus Y$, $n(\gamma,v)$ is the minimum of $M$ and the number of vertices $u\in V(G)\setminus Y$ such that
$f_\II(u)=v$ and $G,\II,A_I\models \theta(u)$.
\end{itemize}

\begin{lemma}\label{lemma-delvers}
Let $I$ be a finite index set, let $\sigma$ be a counter $I$-signature, and let
$\varphi$ be a global quantifier-free first-order $(I,\sigma)$-sentence.
There exists a signature $\sigma'$ obtained from $\sigma$ by adding unary predicate symbols
and changing the triggers on the counter symbols,
and a global quantifier-free first-order $(I,\sigma')$-sentence $\varphi'$
such that the following claim holds.  For any graph $G$, a $G$-interpretation $\II$ of $\sigma$,
a set $Y\subseteq V(G)$ and a $(G,Y,\sigma,\varphi)$-census $n$, we can in linear time find
a $G[Y]$-interpretation $\II$ of $\sigma's$ such that 
$$G,\II,A_I\models \varphi\text{ if and only if }G[Y],\II',A_I\models \varphi'$$
for every $I$-tuple $A_I$ of subsets of $Y$ with $(G,Y,\II,\sigma,\varphi)$-shadow $n$.
\end{lemma}
\begin{proof}
Let $M$ be the largest integer appearing in the formula $\varphi$ or the triggers of $\sigma$.  
For each counter symbol $\gamma$ and each positive integer $m\le M$,
we add to $\sigma'$ a unary predicate $P_{\gamma,m}$, interpreted in $\II'$
as the set of vertices $v\in Y$ such that $n(\gamma,v)\ge m$.
Each function symbol $g$ is interpreted in $\II'$ by setting
$g_{\II'}(v)=g_{\II}(v)$ if $g_{\II}(v)\in Y$ and $g_{\II}(v)=v$ otherwise.  The predicate symbols
are interpreted in $\II'$ as the restrictions of their intepretations in $\II$ to $Y$.

For a counter symbol $\gamma$ with trigger $(\theta,f)$, we set the trigger of $\gamma$ in $\sigma'$
to be $(\theta',f)$, where $\theta'$ is obtained from $\theta$ by replacing each formula of form $\gamma'(x)\ge m$
by the formula
$$(\gamma'(x)\ge m)\lor (P_{\gamma',1}(x)\land \gamma'(x)\ge m-1)\lor\ldots\lor P_{\gamma',m}(x).$$
This ensures that for each $v\in Y$ and each $I$-tuple $A_I$ of subsets of $Y$ with $(G,Y,\II,\sigma,\varphi)$-shadow $n$, we have
$$\min(M,\gamma(\II,A_I,v))=\min(M,\gamma(\II',A_I,v)+n(\gamma,v)),$$
where $n(\gamma,v)=0$ if all neighbors of $v$ belong to $Y$.
For each subformula of $\varphi$ of form $\#\theta\ge m$, let $\theta'$ be obtained by performing the same
replacements.  The formula $\varphi'$ is obtained from $\varphi$ by replacing each such subformula
by the formula $\#\theta'\ge \max(0, m - n(\theta))$.
The fact that $G,\II,A_I\models \varphi$ if and only if $G[Y],\II',A_I\models \varphi'$ holds
for every $I$-tuple $A_I$ of subsets of $Y$ with $(G,Y,\II,\sigma,\varphi)$-shadow $n$
is clear from the construction.
\end{proof}

We also need the fact that counters can be eliminated at the expense of re-introducing quantifiers.

\begin{lemma}\label{lemma-uncount}
Let $I$ be a finite index set, let $\sigma$ be a counter $I$-signature, and
let $\sigma'$ be the simple signature consisting of the predicate and function symbols from $\sigma$.
For any counter symbol $\gamma\in \sigma$ and a positive integer $m$,
there exists a first-order $(I,\sigma')$-formula $\psi_{\gamma,m}$ with one free variable $x$ such that
the following claim holds:
Let $G$ be a graph and let $\II$ be a $G$-interpretation of $\sigma$.
For any $I$-tuple $A_I$ of subsets of vertices of $G$, we have $\gamma(\II,A_I,v)\ge m$ if and only of
$G,\II,A_I\models \psi_{\gamma,m}(v)$.
\end{lemma}
\begin{proof}
We prove the claim by induction along the linear ordering of the counter symbols in $\sigma$,
and thus we can assume that the claim holds for all counter symbols appearing in the trigger $(\theta,f)$ of $\gamma$.
Let $\theta'$ be the $(I,\sigma')$-formula obtained from $\theta$ by replacing each formula of form $\gamma'(x)\ge m'$ by
the formula $\psi_{\gamma',m'}(x)$.  We let $\psi_{\gamma,m}$ be the formula
$$(\exists x_1)\ldots (\exists x_m) \left(\bigwedge_{i<j} x_i\neq x_j\right)\land \bigwedge_{i=1}^m (x_i\neq x\land f(x_i)=x\land \theta'(x_i)).$$
\end{proof}

We now straightforwardly compose the results.
\begin{proof}[Proof of Theorem~\ref{thm-local}]
Let $\sigma_1$ be the counter $I$-signature and $\varphi_1$ the global quantifier-free first-order $(I,\sigma_1)$-sentence
obtained using Theorem~\ref{thm-elim}.  Let $s=2^{\ell(\sigma_1)}$.
Let $\sigma_2$ and $\varphi_2$ be the counter $I$-signature and the global quantifier-free first-order $(I,\sigma_2)$-sentence
obtained using Lemma~\ref{lemma-delvers} for $\sigma_1$ and $\varphi_1$.
We let $\sigma$ be the simple signature consisting of the predicate and function symbols from $\sigma_2$,
and for each counter symbol $\gamma$ of $\sigma_2$ and each positive integer $m$, we let $\psi_{\gamma,m}$
be the formula constructed in Lemma~\ref{lemma-uncount}.  We let $\varphi'$ be the formula obtained from $\varphi_2$
by replacing every subformula of form $\gamma(x)\ge m$ by the formula $\psi_{\gamma,m}(x)$,
and every subformula of form $\#\theta\ge m$ by the formula
$$(\exists x_1)\ldots (\exists x_m) \left(\bigwedge_{i<j} x_i\neq x_j\right)\land \bigwedge_{i=1}^m \theta(x_i).$$

Now, given a graph $G$, we first use the algorithm from Theorem~\ref{thm-elim} to compute a $G$-interpretation $\II_1$ of $\sigma_1$ such that
$$G,A_I\models \varphi\text{ if and only if }G,\II_1,A_I\models \varphi_1$$
for every $I$-tuple $A_I$ of subsets of $V(G)$.  We let $h$ be the $s$-shroud $h_{\II_1}$ and let $X$ be the $h$-center of $Y$.
Let $n$ be the $(G,Y,\II_1,\sigma_1,\varphi_1)$-shadow of the $I$-tuple of empty sets.
Note that by Lemma~\ref{lemma-shroud}, every $I$-tuple $A_I$ of subsets of $X$ has $(G,Y,\II_1,\sigma_1,\varphi_1)$-shadow $n$.
Using the algorithm from Lemma~\ref{lemma-delvers}, we compute a $G[Y]$-interpretation $\II_2$ of $\sigma_2$
such that
$$G,\II_1,A_I\models \varphi_1\text{ if and only if }G[Y],\II_2,A_I\models \varphi_2$$
for every $I$-tuple $A_I$ of subsets of $X$.  By construction, $\varphi'$ is a first-order $(I,\sigma)$-sentence
such that 
$$G[Y],\II_2,A_I\models \varphi_2\text{ if and only if } G[Y],\II_2,A_I\models \varphi'$$
for every $I$-tuple $A_I$ of subsets of $Y$.
Therefore,
$$G,A_I\models \varphi\text{ if and only if }G[Y],\II_2,A_I\models \varphi'$$
for every $I$-tuple $A_I$ of subsets of the $h$-center $X$ of $Y$, and we can
set $\II_Y=\II_2$.
\end{proof}

Next, we prove Theorem~\ref{thm-optfo}.
Let $I$ be a finite index set, let $\sigma$ be a counter $I$-signature,
let $\varphi$ be a global quantifier-free first-order $(I,\sigma)$-sentence,
let $G$ be a graph, and let $\II$ be a $G$-interpretation of $\sigma$.  Suppose $G=L\cup R$ for some subgraphs $L$ and $R$ of $G$.  In this context,
for $I$-tuples $A_I$ and $A'_I$ of subsets of vertices of $L$, we write $A_I\equiv_{(L,R)} A'_I$ if
for every $I$-tuple $B_I$ of subsets of $V(R)\setminus V(L)$, we have
$$G,\II,A_I\cup B_I\models \varphi\text{ if and only if }G,\II,A'_I\cup B_I\models \varphi.$$
The algorithm is based on the following key fact.

\begin{lemma}\label{lemma-numstates}
Let $I$ be a finite index set, let $\sigma$ be a counter $I$-signature,
let $\varphi$ be a global quantifier-free first-order $(I,\sigma)$-sentence,
let $G=L\cup R$ be a graph and let $\II$ be a $G$-interpretation of $\sigma$. Let $h_\II$ be the corresponding
$2^{\ell(\sigma)}$-shroud and let $S=h_{\II}(V(L\cap R))\cap V(L)$.  If $A_I$ and $A'_I$ are $I$-tuples of subsets of $V(L)$ such that
\begin{itemize}
\item $A_I\cap S=A'_I\cap S$ and
\item $A_I$ and $A'_I$ have the same $(G,V(L)\setminus S,\II,\sigma,\varphi)$-shadow $n$,
\end{itemize}
then $A_I\equiv_{(L,R)} A'_I$.
\end{lemma}
\begin{proof}
By the definition of $h_{\II}$, we have $S=h_{\II}(V(R))\cap V(L)$.  Consider any $I$-tuple $B_I$ of subsets of $V(R)\setminus V(L)$.
By Lemma~\ref{lemma-shroud} applied with $X=V(R)\setminus V(L)$,
we have $\gamma(\II,A_I\cup B_I,v)=\gamma(\II,A_I,v)$ and $\gamma(\II,A'_I\cup B_I,v)=\gamma(\II,A'_I,v)$
for each $v\in V(L)\setminus S$.  Consequently, the $(G,V(L)\setminus S,\II,\sigma,\varphi)$-shadow
of both $A_I\cup B_I$ and $A'_I\cup B_I$ is equal to $n$.  Let $\sigma'$, $\varphi'$, and $\II'$ be obtained using Lemma~\ref{lemma-delvers}
for the census $n$ and $Y=V(R)\cup S$.  By the assumptions, we have $(A_I\cup B_I)\cap (V(R)\cup S)=(A'_I\cup B_I)\cap (V(R)\cup S)$, and thus
\begin{align*}
G,\II,A_I\cup B_I&\models \varphi\text{ if and only if}\\
G[V(R)\cup S],\II',(A_I\cup B_I)\cap (V(R)\cup S)&\models \varphi'\text{ if and only if}\\
G[V(R)\cup S],\II',(A'_I\cup B_I)\cap (V(R)\cup S)&\models \varphi'\text{ if and only if}\\
G,\II,A'_I\cup B_I&\models \varphi,
\end{align*}
as required.
\end{proof}

\begin{corollary}\label{cor-numstates}
Let $I$ be a finite index set, let $\sigma$ be a counter $I$-signature,
let $\varphi$ be a global quantifier-free first-order $(I,\sigma)$-sentence,
let $G=L\cup R$ be a graph and let $\II$ be a $G$-interpretation of $\sigma$.
Then $\equiv_{L,R}$ has $\exp(O(|V(L\cap R)|))$ equivalence classes.
\end{corollary}
\begin{proof}
Let $M$ be the largest integer appearing in the formula $\varphi$ and the triggers of $\sigma$
and let $a$ be the number of $x$-local formulas appearing in $\varphi$.
Let $$c=(M+1)^{\ell(\sigma)2^{\ell(\sigma)}}2^{|I|2^{\ell(\sigma)}}.$$
Let $h_\II$ be the $2^{\ell(\sigma)}$-shroud corresponding to $\II$ and let $S=h_{\II}(V(L\cap R))\cap V(L)$.

By Lemma~\ref{lemma-numstates}, each equivalence class of $\equiv_{(L,R)}$ is determined by
\begin{itemize}
\item the $(G,V(L)\setminus S,\sigma,\varphi)$-census $n$: since only the vertices of $S$ can have neighbors in $V(L)\setminus S$, the number of such censuses is
at most $(M+1)^{a+\ell(\sigma)|S|}\le (M+1)^{a+\ell(\sigma)2^{\ell(\sigma)}\cdot |V(L\cap R)|}$; and
\item the restriction of the $I$-tuple to $S$: there are $2^{|I|\cdot|S|}\le 2^{|I|2^{\ell(\sigma)}\cdot |V(L\cap R)|}$
options.
\end{itemize}
Hence, the number of equivalence classes of $\equiv_{L,R}$ is at most $(M+1)^ac^{|V(L\cap R)|}$.
\end{proof}

\begin{proof}[Proof of Theorem~\ref{thm-optfo}]
Let $T$ be the tree of the tree decomposition $\tau$, rooted arbitrarily.  For each node $x\in V(T)$,
let $T_x$ denote the subtree of $T$ induced by $x$ and its descendants.  Let $L_x=G[\bigcup_{y\in V(T_x)} \tau(y)]$
and $R_x=G[\bigcup_{y\in (V(T)\setminus V(T_x))\cup \{x\}} \tau(y)]$, so that $L_x\cup R_x=G$ and $|V(L_x\cap R_x)|=|\tau(x)|\le t+1$.

Let $\sigma$, $\varphi'$ and $\II$ be obtained using Theorem~\ref{thm-elim}.
We use the standard dynamic programming approach, computing for each $x\in V(T)$ a table assigning to each equivalence class $C$
of $\equiv_{L_x,R_x}$ an $I$-tuple $A_{I,C}\in C$ of subsets of $X\cap V(L_x)$ such that $w(A_{I,C})$ is maximized.
Since $\equiv_{L_x,R_x}$ has $\exp(O(t))$ equivalence classes by Corollary~\ref{cor-numstates}, this
can be done in total time $\exp(O(t))|V(G)|$ for all nodes of $\tau$.

Let $r$ be the root of $T$, and note that $L_r=G$ and $R_r=G[\tau(r)]$.  We go over the equivalence classes $C$ of
$\equiv_{L_r,R_r}$ corresponding to $I$-tuples satisfying the property expressed by $\varphi'$ (and thus also by $\varphi$),
and return the $I$-tuple $A_{I,C}$ maximizing $w(A_{I,C})$.
\end{proof}

To finish the argument, we need to prove Theorem~\ref{thm-elim}.

\section{Eliminating the quantifiers}

First, let us note that we can restrict ourselves to non-graph formulas (i.e., formulas
that do not use the adjacency predicate).
\begin{lemma}\label{lemma-elimedge}
Let $I$ be a finite index set, let $\varphi$ be a first-order graph $I$-sentence,
and let $\GG$ be a class of graphs with bounded expansion.  There exists a simple signature $\sigma$
and a first-order $(I,\sigma)$-sentence $\varphi'$ such that the following claim holds.
Given a graph $G\in \GG$, we can in linear time compute a $G$-interpretation $\II$ of $\sigma$ such that
$$G,A_I\models \varphi\text{ if and only if }G,\II,A_I\models \varphi'$$
for every $I$-tuple $A_I$ of subsets of $V(G)$.
\end{lemma}
\begin{proof}
Let $h$ be the function bounding the expansion of $\GG$ and let $d=2h(0)$.  We let $\sigma$ consist of
the function symbols $f_1$, \ldots, $f_d$, and we obtain $\varphi'$ from $\varphi$ by
replacing each occurrence of the edge predicate of form $E(t_1,t_2)$ by the formula
$$t_1\neq t_2\land \bigvee_{i=1}^d f_i(t_1)=t_2\lor f_i(t_2)=t_1.$$
Note that any subgraph of a graph $G\in\GG$ has average degree at most $2\nabla_0(G)\le d$,
and thus $G$ is $d$-degenerate.  Hence, we can in linear time find an acyclic orientation $\vec{G}$ of
$G$ such that each vertex has outdegree at most $d$.  Let $\kappa:E(G)\to\{1,\ldots,d\}$ be chosen arbitrarily
so that $\kappa(e_1)\neq \kappa(e_2)$ for any distinct edges $e_1$ and $e_2$ oriented in $\vec{G}$ away
from a common incident vertex. We interpret the function $f_i$ in $\II$ by setting $f_i(u)=v$ if
$v$ is an outneighbor of $u$ in $\vec{G}$ and $\kappa(uv)=i$, and $f_i(u)=u$ if no such vertex $v$ exists.
\end{proof}

Next, we need some definitions and auxiliary results.  Recall that in a rooted forest, the \emph{ancestors} of a vertex $v$
are the vertices on the path from $v$ to a root (including $v$) and \emph{descendants} of $v$ are the vertices for which $v$ is
an ancestor.  \emph{Strict} ancestors and descendants of $v$ are the ancestors and descendants different from $v$.
If $G$ is a graph and $F$ is a rooted forest with $V(F)\subseteq V(G)$, we say that $F$ is a \emph{scaffolding} of $G$ if
$F\subseteq G$ and for each edge $uv\in E(G)$, if $u,v\in V(F)$, then $u$ is a strict ancestor or descendant of $v$ in $F$.
For a positive integer $s$, a system $\ZZ$ of scaffoldings in $G$ is \emph{$s$-generic}
if for every set $X\subseteq V(G)$ of size at most $s$, there exists a scaffolding $F\in \ZZ$
such that $X\subseteq V(F)$.  We use the following consequence of Corollary~\ref{cor-lowtd}.

\begin{lemma}\label{lemma-scaffold}
For every class $\GG$ of graphs with bounded expansion and every positive integer $s$, there exists a positive integer $c$ and a linear-time
algorithm that, given a graph $G\in \GG$, returns an $s$-generic system of size $c$ of scaffoldings of depth at most $2^s$.
\end{lemma}
\begin{proof}
Let $\CC$ be the $(s,1/c)$-generic cover of size $c$ and treedepth at most $s$ obtained using Corollary~\ref{cor-lowtd}.
For each $C\in\CC$, the graph $G[C]$ has treedepth at most $s$.  Since a path with $2^s$ vertices has treedepth
greater than $s$, we conclude that $P_{2^s}\not\subseteq G[C]$.  Let $F_C$ be an arbitrary DFS forest of $G[C]$;
then $F_C$ has depth less than $2^s$.  Moreover, $F_C\subseteq G$ and $G[C]$ is a subgraph of the closure of $F_C$, and
thus $F_C$ is a scaffolding.  Hence, we can return the system $\{F_C:C\in\CC\}$.
\end{proof}

The \emph{depth} of a vertex $v$ in a scaffolding $F$ is the number of edges on the path from $v$ to a root;
in particular, the roots of $F$ have depth $0$.
A \emph{representation} $\RR_F$ of a scaffolding $F$ in a graph $G$ consists of a unary predicate symbol
$\cin{F}$ interpreted by the set $V(F)$ and a unary function symbol $\prt{F}$ interpreted as the function defined so that
if $v$ is a non-root vertex of $F$, then $\prt{F}(v)$ is the parent of $v$ in $F$, and $\prt{F}(v)=v$ otherwise.

A \emph{term} is a variable or a composition of unary function symbols applied to a variable.
For a finite set $X$ of terms, an \emph{$X$-template} is a pair $(Q,\mu)$, where $Q$ is a rooted forest and $\mu:X\to V(Q)$ is a
function such that
\begin{itemize}
\item every leaf of $Q$ is in the image of $\mu$, and
\item for every function symbol $f$ and a term $t$, if $t\in X$ and $f(t)\in X$, then $\mu(t)$ is an ancestor or a descendant
of $\mu(f(t))$ in $Q$.
\end{itemize}
Note that for each $d$, there are only finitely many non-isomorphic $X$-templates of depth at most $d$.
We refer to the last condition as \emph{guard-consistency}, and note it implies the following useful fact.
\begin{observation}\label{obs-guard}
Let $X$ be a finite set of terms, containing a term $t$ with free variable $z$ and all its subterms.
Let $T=(Q,\mu)$ be an $X$-template.   Then there exists a subterm $t'$ of $t$ such that $\mu(t')$ is
a common ancestor of $\mu(t)$ and $\mu(z)$ in $Q$.  In particular, if $\mu$ does not map any term containing $z$ to a strict ancestor
of $\mu(t)$, then $\mu(z)$ is a descendant of $\mu(t)$.
\end{observation}
\begin{proof}
By the guard-consistency, $\mu$ maps the subterms of $t$ to a walk in the closure of $Q$ from $\mu(z)$ to $\mu(t)$,
and such a walk must pass through a common ancestor of these vertices.
\end{proof}
If $X$ is a set of terms, $T=(Q,\mu)$ is an $X$-template,
and $F$ is a symbol for a scaffolding, then let $\tau_{T,F}$ be the conjunction of the following formulas (where $f^{(k)}$ denotes the $k$-times
nested application of the function $f$, and in particular $f^{(0)}(t)$ is $t$):
\begin{itemize}
\item $\cin{F}(t)$ for each $t\in X$,
\item for $t\in X$ such that the depth of $\mu(t)$ in $Q$ is $k$:
\begin{itemize}
\item $\prt{F}^{(i)}(t)\neq \prt{F}^{(i+1)}(t)$ for $i\in \{0,\ldots,k-1\}$ and
\item $\prt{F}^{(k)}(t)=\prt{F}^{(k+1)}(t)$,
\end{itemize}
\item for each $t,t'\in X$ where $\mu(t)$ and $\mu(t')$ are at depths $k_1$ and $k_2$, respectively, in the same component of $Q$,
and their nearest common ancestor is at depth $k$:
\begin{itemize}
\item $\prt{F}^{(k_1-i)}(t)\neq \prt{F}^{(k_2-i)}(t')$ for $i\in \{k+1,\ldots,\min(k_1,k_2)\}$ and
\item $\prt{F}^{(k_1-k)}(t)=\prt{F}^{(k_2-k)}(t')$,
\end{itemize}
\item $\prt{F}^{(k_1)}(t)\neq \prt{F}^{(k_2)}(t')$ for each $t,t'\in X$ where $\mu(t)$ and $\mu(t')$ are at depths $k_1$ and $k_2$,
respectively, in different components of $Q$.
\end{itemize}
Consider an assignment $\omega$ of vertices of $G$ to the variables appearing in $X$.
We say that $\omega$ \emph{$F$-matches} an $X$-template $(Q,\mu)$ (with respect to a fixed $G$-interpretation $\II$
of the unary functions appearing in the terms of $X$) if there exists an injective homomorphism $h$ from $Q$ to $F$ mapping the roots of
$Q$ to roots of $F$ and such that for each term $t\in X$, the value of $t$ according to $\II$ and $\omega$ is
equal to $h(\mu(t))$.  The formula $\tau_{T,F}$ is constructed so that the following claim holds.

\begin{observation}\label{obs-templ}
Let $X$ be a finite set of terms.  Let $G$ be a graph, let $\II$ be a $G$-interpretation of the function symbols appearing
in the terms of $X$, and let $F$ be a scaffolding in $G$.
An assignment $\omega$ of vertices of $G$ to the variables appearing in $X$ $F$-matches an $X$-template $T$ if and only if
$\II,\RR_F\models \tau_{T,F}(\omega)$.
\end{observation}

For a finite index set $I$ and counter $I$-signatures $\sigma$ and $\sigma'$, we write
$\sigma\subseteq \sigma'$ if all symbols from $\sigma$ belong to $\sigma'$, the
counter symbols of $\sigma$ have the same triggers and ordering in $\sigma'$,
and they are smaller than all the counter symbols from $\sigma'\setminus\sigma$.
The \emph{term set} of a formula $\varphi$ is the set $X$ of all terms appearing in $\varphi$,
including the subterms (e.g., if the term $f(g(x))$ appears in $\varphi$, we also include $g(x)$ and $x$ in $X$).
Let us now prove a key lemma that enables us to eliminate a quantifier subject to a template restriction.

\begin{lemma}\label{lemma-elimone}
Let $I$ be a finite index set, let $\sigma$ be a counter $I$-signature, and let
$\psi$ be a quantifier-free first-order $(I,\sigma)$-formula with one free variable $z$, where $\psi$ is 
a conjunction of elementary formulas and their negations and does not use the equality relation symbol.
Let $T=(Q,\mu)$ be an $X$-template for the term set $X$ of $\psi$ and let $r$ be a vertex of $Q$
such that $\mu$ maps all terms of $X$ to descendants of $r$.
There exists a counter $I$-signature $\sigma'\supseteq\sigma$ and an $x$-local
$(I,\sigma')$-formula $\lambda$ such that the following claim holds.
For any graph $G$, any $G$-interpretation $\II$ of $\sigma$,
and any scaffolding $F$ in $G$ of depth at most $d$, we can in linear time compute a $G$-interpretation $\II'$ of
$\sigma'$ such that for any $v\in V(G)$ and any $I$-tuple $A_I$ of subsets of $V(G)$,
\begin{itemize}
\item $G,\II',A_I\models\lambda(v)$ if and only if
\item $v\in V(F)$ and the depth of $v$ in $F$ is equal to the depth of $r$ in $Q$ and there exists a descendant $v'$ of $v$ in $F$
such that the assignment of $v'$ to $z$ $F$-matches $T$ and $G,\II,A_I\models\psi(v')$.
\end{itemize}
\end{lemma}
\begin{proof}
We include in $\sigma'$ the symbols of the representation of $F$ and in $\II'$ the representation $\RR_F$.
We add to $\sigma'$ a counter $A_{u,y}$ for each strict descendant $y$ of $r$ and its parent $u$ in $Q$.
These counters are ordered in a non-increasing order according to the depth of $u$.
The trigger for $A_{u,y}$ is $(\prt{F},\lambda_y)$, where $\lambda_y$ is the conjunction of $x$-local formulas obtained as follows.
\begin{itemize}
\item[(i)] If $y$ is not an ancestor of $\mu(z)$, then let $t_y\in X$ be an arbitrarily chosen
term such that $y$ is an ancestor of $\mu(t_y)$.  By Observation~\ref{obs-guard},
there exists a subterm $t'_y$ of $t_y$ such that $\mu(t'_y)$ is a common ancestor of $\mu(t_y)$ and $\mu(z)$.
Since $y$ is not an ancestor of $\mu(z)$, we conclude that $\mu(t'_y)$ is an ancestor of $u$.
Since $t'_y$ is a subterm of $t_y$, we have $t_y=a_y(t'_y)$, where $a_y$ is a composition of function applications. 
We add to the conjunction $\lambda_y$ the formula $p_y(x)$, where $p_y$ is a new unary predicate symbol interpreted in $\II'$
as the set of vertices $w\in V(F)$ with the depth in $F$ equal to the depth of $y$ in $Q$
such that for the ancestor $w'$ of $w$ in $F$ whose depth is equal to the depth of $\mu(t'_y)$ in $Q$,
the vertex $a_y(w')$ (according to the interpretation $\II'$) is a descendant of $w$.

Note that this ensures that the vertex assigned to $x$ is at the same depth as $y$ and that its
parent has at most one child for which the formula $\lambda_y$ holds.
\item[(ii)] If $y=\mu(z)$, then we add to the conjunction $\lambda_y$ the formula $\tau(x)$, where
$\tau$ is a new unary predicate interpreted as the
set of vertices $w\in V(F)$ such that the assignment of $w$ to $z$ $F$-matches $T$.
\item[(iii)] We add to the conjunction $\lambda_y$ the formula $A_{y,y'}(x)\ge 1$ for each child $y'$ of $y$ in $Q$.
\item[(iv)] For each operand of the conjunction $\psi$ of form $p(t)$ ($p$ is a unary predicate symbol,
a comparison of the value of a counter, or a negation of one of these) where $t$ is a term such that $\mu(t)=y$,
we add to the conjunction $\lambda_y$ the formula $p(x)$.
\end{itemize}
We define the formula $\lambda_r$ in the same way (note that (i) does not apply, since $r$ is an ancestor of $\mu(z)$)
and let $\lambda\equiv\lambda_r$.  Let us show that $\lambda$ satisfies the forwards implication from the statement of the lemma
(it is easy to see that the backwards one holds).

Suppose that $G,\II',A_I\models\lambda(v)$.  Let $r=y_0,y_1,\ldots,y_d=\mu(z)$ be the path in $Q$ from $r$ to $\mu(z)$.
Let $v_0=v$. For $i=1,\ldots,d$, note that since $G,\II',A_I\models\lambda_{y_{i-1}}(v_{i-1})$, (iii)
ensures that $A_{y_{i-1},y_i}(v_{i-1})>0$, and since the trigger of $A_{y_{i-1},y_i}$ is $(\prt{F},\lambda_{y_i})$,
this implies that there exists a child $v_i$ of $v_{i-1}$ in $F$ such that $G,\II',A_I\models\lambda_{y_i}(v_i)$; we fix
any such child $v_i$.

Let $v'=v_d$.  Since $y_d=\mu(z)$ and $G,\II',A_I\models\lambda_{y_d}(v')$, (ii) implies that
the assignment of $v'$ to $z$ $F$-matches $T$, and in particular $v'\in V(F)$ has the same depth in $F$
as $\mu(z)$ in $Q$.  Note that $v_0=v$ is a vertex of $F$ whose depth is by $d$ smaller than the depth of $v'$, and thus equal
to the depth of $r$ in $Q$.

Finally, we need to argue that $G,\II,A_I\models\psi(v')$.  Consider any operand $p(t)$ of the conjunction $\psi$,
and let $y_a$ (where $a\in\{0,\ldots,d\}$) be the nearest common ancestor of $\mu(t)$ and $\mu(z)$ in $Q$.
Let $y_a=y'_0,y'_1,\ldots,y'_c=\mu(t)$ be the path in $Q$ from $y_a$ to $\mu(t)$.  Since the assignment of $v'$ to $z$ $F$-matches $T$,
$F$ contains a path $v_a=w_0,\ldots,w_c$ from $v_a$ to the vertex $w_c=t(v')$ (evaluated according to $\II$) of the same length.
For $i=1,\ldots,c$, note that since $G,\II',A_I\models\lambda_{y'_{i-1}}(w_{i-1})$, (iii)
ensures that $A_{y'_{i-1},y'_i}(w_{i-1})>0$, and since the trigger of $A_{y'_{i-1},y'_i}$ is $(\prt{F},\lambda_{y'_i})$,
this implies that there exists a child $w'_i$ of $w_{i-1}$ in $F$ such that $G,\II',A_I\models\lambda_{y'_i}(w'_i)$.
Note that (i) ensures that there is only one such child, and since the assignment of $v'$ to $z$ $F$-matches $T$,
we conclude that $w_i=w'_i$.  Consequently, we have $G,\II',A_I\models\lambda_{y'_c}(w_c)$, and since $y'_c=\mu(t)$
and $w_c=t(v')$, (iv) implies $G,\II,A_I\models p(t(v'))$.
\end{proof}

Using Lemma~\ref{lemma-elimone}, we now prove the following strengthening.
\begin{lemma}\label{lemma-elimtemp}
Let $I$ be a finite index set, let $\sigma$ be a counter $I$-signature, and let
$\psi$ be a quantifier-free first-order $(I,\sigma)$-formula with free variables $W$, where $\psi$ is 
a conjunction of elementary formulas and their negations, and let $z\in W$ be a free variable.
Let $T=(Q,\mu)$ be an $X$-template for the term set $X$ of $\psi$.
There exists a counter $I$-signature $\sigma'\supseteq\sigma$ and a global quantifier-free first-order
$(I,\sigma)$-formula $\psi'$ with free variables $W\setminus\{z\}$
such that the following claim holds.  For any graph $G$, any $G$-interpretation $\II$ of $\sigma$,
and any scaffolding $F$ in $G$, we can in linear time compute a $G$-interpretation $\II'$ of
$\sigma'$ such that for any assignment $\omega':W\setminus\{z\}\to V(G)$ and any $I$-tuple $A_I$ of subsets of $V(G)$,
\begin{itemize}
\item[(a)] $\omega'$ extends to an assignment $\omega:W\to V(G)$ $F$-matching $T$ such that $G,\II,A_I\models\psi(\omega)$ if and only if
\item[(b)] $G,\II',A_I\models\psi'(\omega')$.
\end{itemize}
\end{lemma}
\begin{proof}
If the conjunction $\psi$ contains an elementary formula of form $t=t'$ and $\mu(t)\neq \mu(t')$ or
an elementary formula of form $\lnot(t=t')$ and $\mu(t)=\mu(t')$, then $\psi$ is false
for any assignment $\omega$ that $F$-matches $T$ and we can set $\psi'=\text{false}$.
Similarly, if $\psi$ contains an elementary formula of form $t=t'$ and $\mu(t)=\mu(t')$
or $\lnot(t=t')$ and $\mu(t)\neq\mu(t')$, then this subformula is true for any
assignment $\omega:W\to V(G)$ $F$-matching $T$, and we can delete the subformula from $\psi$ without
affecting the validity of (a).  Hence, we can without loss of generality assume that $\psi$ does not
use equality.

We include in $\sigma'$ the symbols of the representation of $F$ and in $\II'$ the representation $\RR_F$.
If there exists a term $t\in X$ not containing the variable $z$ such that $\mu(z)$ is an ancestor of $\mu(t)$ in $Q$,
then, letting $k$ be the length of the path between $\mu(t)$ and $\mu(z)$ in $Q$,
we can set $\psi'$ to be the formula obtained from $\tau_{T,F}\land \psi$ by substituting
$\prt{F}^{(k)}(t)$ for $z$.  Hence, suppose this is not the case, and thus $\mu$ maps only terms containing $z$ to
descendants of $\mu(z)$.

Let $\psi_1$ be the formula obtained from $\psi$ by repeatedly performing the following operation as long as there exist terms
$t$ and $t'$ appearing in the formula such that $t'$ contains the variable $z$, $t$ does not contain $z$, and
$\mu(t')$ is an ancestor of $\mu(t)$ in $Q$.  Without loss of generality, choose such a term $t'$ for which the
length $k$ of the path from $\mu(t)$ to $\mu(t')$ in $Q$ is maximum, and thus $\mu$ does not map terms containing $z$ to
strict ancestors of $\mu(t')$.  By Observation~\ref{obs-guard}, this implies $\mu(t')$ is a strict ancestor of $\mu(z)$ in $Q$;
let $k'$ be the length of the path between $\mu(z)$ and $\mu(t')$ in $Q$.  Let us add to $\sigma'$ a new unary predicate symbol $q$,
interpreted in $\II'$ as the set of vertices $v\in V(F)$ such that the vertex $t'_{\II}(v)$ is the ancestor of $v$ in $F$ at distance $k'$.
We replace all occurrences of the term $t'$ (even as subterms of other terms) in the formula
by the term $\prt{F}^{(k)}(t)$, and add the operand $q(z)$ to the conjunction.

Let $X_1$ be the term set of $\psi_1$.  Observe that along with the replacements described in the previous paragraph, we
can transform the template $T$ into an $X_1$-template $T_1=(Q,\mu_1)$ such that for any assignment $\omega':W\setminus\{z\}\to V(G)$,
\begin{itemize}
\item $\omega'$ extends to an assignment $\omega:W\to V(G)$ such that $\omega$ $F$-matches $T$ and $G,\II,A_I\models\psi(\omega)$ if and only if
\item $\omega'$ extends to an assignment $\omega:W\to V(G)$ such that $\omega$ $F$-matches $T_1$ and $G,\II',A_I\models\psi_1(\omega)$.
\end{itemize}
Moreover, $\mu_1$ also maps only terms containing $z$ to descendants $\mu_1(z)$.
Let $r$ be the vertex of $Q$ of the smallest depth such that $\mu(z_1)$ is a descendant of $r$
and $\mu_1$ maps only terms containing $z$ to descendants of $r$.  Let $B$ be the set consisting of 
descendants of $r$ in $Q$.

We claim that $\mu_1$ maps all terms of $X_1$ containing $z$ to $B$.  Indeed, consider
for a contradiction a term $t'$ containing $z$ such that $\mu_1(t')\not\in B$
and $t'$ has minimum depth.  By Observation~\ref{obs-guard}, $\mu_1(t')$ is a strict ancestor of $\mu_1(z)$, and thus
also a strict ancestor of $r$.  However, by the choice of $r$, $\mu_1$ maps a term $t$ not containing $z$ to
a descendant of $\mu_1(t')$; and we would have replaced $t'$ in the construction of the formula $\psi_1$.

Let $X_2$ be the set of terms in $X_1$ containing $z$, let $\mu_2$ be the restriction of $\mu_1$ to $X_2$,
and let $Q_2$ be the rooted tree induced in $Q$ by $B$ and the ancestors of $r$.
Then $T_2=(Q_2,\mu_2)$ is an $X_2$-template.  Let $Q_3=Q-B$ and let $\mu_3$ be the restriction of $\mu_1$ to $X_3=X_1\setminus X_2$; then
$T_3=(Q_3,\mu_3)$ is an $X_3$-template.  If $r$ is a root of $Q$, then
let $r_1$, \ldots, $r_m$ be all roots of $Q$ distinct from $r$.  Otherwise, let $r_1, \ldots, r_m$ be the
vertices of $Q$ distinct from $r$ with the same parent as $r$.
For $i\in \{1,\ldots,m\}$, let $s_i$ be a term $\prt{F}^{(k_i)}(s'_i)$, where $s'_i\in X_3$ is
a term such that $\mu'(s'_i)$ is a descendant of $r_i$ in $Q$ at distance $k_i$.
In the case that $r$ is not a root, let us additionally define a term $s$ as follows.
If there exists a term in $X_3$ mapped to the parent of $r$ by $\mu_3$, then we choose any such term as $s$.
Otherwise, the choice of $r$ implies $m\ge 1$ and we let $s=\prt{F}(s_1)$. 

Let $\psi_2$ be the conjunction of the operands of $\psi_1$ containing $z$ and let $\psi_3$ be the conjunction of the operands not containing $z$. 
Since $\psi$ (and thus also $\psi_1$) does not contain operands of form $t=t'$ and $\lnot (t=t')$ (and does not use the edge predicate),
$z$ is the only variable appearing in $\psi_2$.  Note that for any assignment $\omega':W\setminus\{z\}\to V(G)$,
\begin{itemize}
\item $\omega'$ extends to an assignment $\omega:W\to V(G)$ such that $\omega$ $F$-matches $T_1$ and $G,\II',A_I\models\psi_1(\omega)$ if and only if
\item $\omega'$ $F$-matches $T_3$ and $G,\II',A_I\models \psi_3(\omega')$ and there exists a vertex $v\in V(F)$
whose depth in $F$ is equal to the depth of $r$ in $Q$ such that
\begin{itemize}
\item $v$ has a descendant $v'$ such that the assignment of $v'$ to $z$ $F_2$-matches $T_2$ and $G,\II',A_I\models\psi_2(v')$,
\item $v$ is different from $s_1$, \ldots, $s_m$ interpreted according to $\II'$ and $\omega'$, and
\item if $r$ is not a root, then the parent of $v$ in $F$ is $s$ interpreted according to $\II'$ and $\omega'$.
\end{itemize}
\end{itemize}
We apply Lemma~\ref{lemma-elimone} to $\psi_2$ and $T_2$, extending $\sigma'$ and $\II'$ and obtaining a formula $\lambda$ such that
the above condition is equivalent to
\begin{itemize}
\item $\omega'$ $F$-matches $T_3$ and $G,\II',A_I\models \psi_3(\omega')$ and there exists a vertex $v\in V(F)$
such that 
\begin{itemize}
\item $G,\II',A_I\models \lambda(v)$,
\item $v$ is different from $s_1$, \ldots, $s_m$ interpreted according to $\II'$ and $\omega'$, and
\item if $r$ is not a root, then the parent of $v$ in $F$ is $s$ interpreted according to $\II'$ and $\omega'$.
\end{itemize}
\end{itemize}

In case $r$ is a root of $Q$, we can now set
$$\psi'\equiv \tau_{T_3,F}\land\psi_3\land \bigvee_{J\subseteq \{1,\ldots,m\}} \#\lambda\ge m+1-|J|\land \bigwedge_{j\in J} \lnot \lambda(s_j).$$
If $r$ is not a root of $Q$, we additionally introduce a counter $A$ with trigger $(\prt{F},\lambda)$ and we set
$$\psi'\equiv \tau_{T_3,F}\land\psi_3\land\bigvee_{J\subseteq \{1,\ldots,m\}} A(s)\ge m+1-|J|\land \bigwedge_{j\in J} \lnot \lambda(s_j).$$
\end{proof}

We further strengthen Lemma~\ref{lemma-elimtemp} to the following quantifier elimination lemma.

\begin{lemma}\label{lemma-elim}
Let $I$ be a finite index set, let $\sigma$ be a counter $I$-signature, and let
$\varphi\equiv(Q z)\psi$ be a global first-order $(I,\sigma)$-formula, where $\psi$ is quantifier-free and $Q$ is a quantifier.
Let $\GG$ be a class of graphs with bounded expansion.  There exists a counter $I$-signature $\sigma'\supseteq\sigma$
and a global quantifier-free first-order $(I,\sigma)$-formula $\varphi'$ with the same free variables as $\varphi$
such that the following claim holds.  For any $G\in \GG$ and a $G$-interpretation $\II$ of $\sigma$,
we can in linear time compute a $G$-interpretation $\II'$ of $\sigma'$ such that for any assignment $\omega'$
of vertices of $G$ to variables and any $I$-tuple $A_I$ of subsets of $V(G)$,
$$G,\II,A_I\models\varphi(\omega')\text{ if and only if }G,\II',A_I\models\varphi'(\omega').$$
\end{lemma}
\begin{proof}
If $\varphi\equiv(\forall z)\psi$, we instead consider the formula $\varphi_0\equiv \lnot\varphi\equiv (\exists z) \lnot\psi$.  Processing
this formula as described below, we obtain a formula $\varphi'_0$ such that
$$G,\II,A_I\models\varphi_0(\omega')\text{ if and only if }G,\II',A_I\models\varphi'_0(\omega').$$
Then, we return the formula $\lnot\varphi'_0$.  Hence, we can assume that $\varphi\equiv(\exists z)\psi$.

Next, let us use the induction by the complexity of $\psi$ to simplify the formula.  Without loss of generality, we can assume $\psi$
is in the disjunctive normal form.  If $\psi\equiv\psi_1\lor \ldots\lor \psi_m$ for $m>1$, then note that $\varphi$ is equivalent
to $(\exists z) \psi_1\lor \ldots\lor (\exists z) \psi_m$.  We apply the induction hypothesis to $(\exists z) \psi_1$, \ldots,
$(\exists z) \psi_m$ in order (accumulating new symbols and their interpretations) and finally return the disjunction of the resulting
formulas.

Hence, we can assume $\psi$ is a conjunction of elementary formulas and their negations.
If $\psi\equiv \psi_1\land \psi_2$, where $\psi_1$ is a formula of form $\#\theta\ge n$ for some $x$-local formula $\theta$
or a negation of such a formula, then note that $\varphi$ is equivalent to $\psi_1 \land (\exists z) \psi_2$.
We apply the induction hypothesis to $(\exists z) \psi_2$ and return the conjunction of the resulting formula with $\varphi$.

Hence, we can assume that $\psi$ is a non-global formula.  Let $X$ be the term set of $\psi$.
Let $c$ and $d=2^{|X|}$ be the constants such that by Lemma~\ref{lemma-scaffold}, any graph from $\GG$ admits an $|X|$-generic system
of scaffoldings of depth at most $d$ and size $c$.  Let $F_1$, \ldots, $F_c$ be labels for scaffoldings and let
$T_1,\ldots, T_b$ be all possible $X$-templates of depth at most $d$.  We let $\varphi'$ be the disjunction of formulas
$\varphi_{i,j}$ obtained using Lemma~\ref{lemma-elimtemp} applied to $\psi$, $F_i$ and $T_j$ (accumulating new symbols to $\sigma'$)
for all $i\in\{1,\ldots,c\}$ and $j\in\{1,\ldots,b\}$.

Given a graph $G\in\GG$ and the $G$-interpretation $\II$, we use the algorithm from Lemma~\ref{lemma-scaffold} to obtain
an $|X|$-generic system $\ZZ$ of scaffoldings of $G$ of depth at most $d$ and size $c$, and assign the labels $F_1$, \ldots, $F_c$
to the elements of this system.  We now apply the algorithm from Lemma~\ref{lemma-elimtemp} for all $F_i$ and $T_j$ with
$i\in\{1,\ldots,c\}$ and $j\in\{1,\ldots,b\}$ to obtain the interpretation $\II'$.

Observe that for any any assignment $\omega'$ of vertices of $G$ to free variables, we have $G,\II,A_I\models\varphi(\omega')$
if and only if $\omega'$ extends to an assignment $\omega$ giving a value to $z$ such that $G,\II,A_I\models\psi(\omega)$.
Since the cover $\ZZ$ is $|X|$-generic, this is the case exactly if for some scaffolding $F_i$ and a template $T_j$,
$\omega'$ extends to an assignment $\omega$ $F_i$-matching $T_j$ such that $G,\II,A_I\models\varphi(\omega)$.
Equivalently, this is the case exactly if for some scaffolding $F_i$ and a template $T_j$, $G,\II',A_I\models\varphi_{i,j}(\omega')$,
i.e., if $G,\II',A_I\models\varphi'(\omega')$.
\end{proof}

To conclude the argument, we iterate Lemma~\ref{lemma-elim}.

\begin{proof}[Proof of Theorem~\ref{thm-elim}]
By Lemma~\ref{lemma-elimedge}, we can assume $\varphi$ does not use the adjacency predicate.
Without loss of generality, we assume $\varphi$ is in the prenex normal form, that is, $\varphi\equiv (Q_1x_1)(Q_2x_2)\ldots (Q_nx_n) \psi_n$
for a quantifier-free formula $\psi_n$.  For $i=n,\ldots, 1$, we apply Lemma~\ref{lemma-elim} to the formula
$(Q_ix_i)\psi_i$, obtaining an equivalent formula $\psi_{i-1}$ (and accumulating new symbols and their interpretations).
In the end, we return the formula $\psi_0$.
\end{proof}

\bibliographystyle{siam}
\bibliography{../data.bib}

\begin{thebibliography}{10}

\bibitem{amiri2017distributed}
{\sc S.~Akhoondian~Amiri, P.~Ossona~de Mendez, R.~Rabinovich, and S.~Siebertz},
  {\em Distributed domination on graph classes of bounded expansion}, in
  Proceedings of the 30th on Symposium on Parallelism in Algorithms and
  Architectures, SPAA '18, New York, NY, USA, 2018, Association for Computing
  Machinery, pp.~143--151.

\bibitem{alon1990separator}
{\sc N.~Alon, P.~Seymour, and R.~Thomas}, {\em A separator theorem for graphs
  with an excluded minor and its applications}, in Proceedings of the
  twenty-second annual ACM symposium on Theory of computing, ACM, 1990,
  pp.~293--299.

\bibitem{baker1994approximation}
{\sc B.~Baker}, {\em Approximation algorithms for {NP}-complete problems on
  planar graphs}, Journal of the ACM (JACM), 41 (1994), pp.~153--180.

\bibitem{berman1999some}
{\sc P.~Berman and M.~Karpinski}, {\em On some tighter inapproximability
  results}, in International Colloquium on Automata, Languages, and
  Programming, Springer, 1999, pp.~200--209.

\bibitem{bodlaender1996linear}
{\sc H.~L. Bodlaender}, {\em A linear-time algorithm for finding
  tree-decompositions of small treewidth}, SIAM Journal on computing, 25
  (1996), pp.~1305--1317.

\bibitem{chan2012weighted}
{\sc T.~M. Chan, E.~Grant, J.~K{\"o}nemann, and M.~Sharpe}, {\em Weighted
  capacitated, priority, and geometric set cover via improved quasi-uniform
  sampling}, in Proceedings of the twenty-third annual ACM-SIAM symposium on
  Discrete Algorithms, SIAM, 2012, pp.~1576--1585.

\bibitem{courcelle}
{\sc B.~Courcelle}, {\em The monadic second-order logic of graphs. {I}.
  {R}ecognizable sets of finite graphs}, Information and computation, 85
  (1990), pp.~12--75.

\bibitem{dawar2006approximation}
{\sc A.~Dawar, M.~Grohe, S.~Kreutzer, and N.~Schweikardt}, {\em Approximation
  schemes for first-order definable optimisation problems}, in 21st Annual IEEE
  Symposium on Logic in Computer Science (LICS'06), IEEE, 2006, pp.~411--420.

\bibitem{demaine2005bidimensionality}
{\sc E.~D. Demaine and M.~Hajiaghayi}, {\em Bidimensionality: new connections
  between {FPT} algorithms and {PTAS}s}, in Proceedings of the sixteenth annual
  ACM-SIAM symposium on Discrete algorithms, Society for Industrial and Applied
  Mathematics, 2005, pp.~590--601.

\bibitem{devospart}
{\sc M.~DeVos, G.~Ding, B.~Oporowski, D.~Sanders, B.~Reed, P.~Seymour, and
  D.~Vertigan}, {\em Excluding any graph as a minor allows a low tree-width
  2-coloring}, J. Comb. Theory, Ser. B, 91 (2004), pp.~25--41.

\bibitem{domker}
{\sc P.~G. Drange, M.~Dregi, F.~V. Fomin, S.~Kreutzer, D.~Lokshtanov,
  M.~Pilipczuk, M.~Pilipczuk, F.~Reidl, F.~S. Villaamil, S.~Saurabh,
  S.~Siebertz, and S.~Sikdar}, {\em {Kernelization and Sparseness: the Case of
  Dominating Set}}, in 33rd Symposium on Theoretical Aspects of Computer
  Science (STACS 2016), vol.~47 of Leibniz International Proceedings in
  Informatics (LIPIcs), Dagstuhl, Germany, 2016, Schloss
  Dagstuhl--Leibniz-Zentrum fuer Informatik, pp.~31:1--31:14.

\bibitem{apxdomin}
{\sc Z.~Dvo{\v{r}}{\'a}k}, {\em Constant-factor approximation of domination
  number in sparse graphs}, European Journal of Combinatorics, 34 (2013),
  pp.~833--840.

\bibitem{twd}
\leavevmode\vrule height 2pt depth -1.6pt width 23pt, {\em Sublinear
  separators, fragility and subexponential expansion}, European Journal of
  Combinatorics, 52 (2016), pp.~103--119.

\bibitem{onsubsep}
\leavevmode\vrule height 2pt depth -1.6pt width 23pt, {\em On classes of graphs
  with strongly sublinear separators}, European Journal of Combinatorics, 71
  (2018), pp.~1--11.

\bibitem{thin}
{\sc Z.~Dvo{\v{r}}{\'a}k}, {\em Thin graph classes and polynomial-time
  approximation schemes}, in Proceedings of the Twenty-Ninth Annual ACM-SIAM
  Symposium on Discrete Algorithms (SODA'18), ACM, 2018, pp.~1685--1701.

\bibitem{apxlp}
{\sc Z.~Dvo{\v{r}}{\'a}k}, {\em On distance $r$-dominating and $2r$-independent
  sets in sparse graphs}, J. Graph Theory, 91 (2019), pp.~162--173.

\bibitem{bakergame}
\leavevmode\vrule height 2pt depth -1.6pt width 23pt, {\em Baker game and
  polynomial-time approximation schemes}, in Proceedings of the Thirty-First
  Annual ACM-SIAM Symposium on Discrete Algorithms, SODA '20, Society for
  Industrial and Applied Mathematics, 2020, pp.~2227--2240.

\bibitem{dvorak2013testing}
{\sc Z.~Dvo{\v{r}}{\'a}k, D.~Kr{\'a}l', and R.~Thomas}, {\em Testing
  first-order properties for subclasses of sparse graphs}, Journal of the ACM
  (JACM), 60 (2013), p.~36.

\bibitem{dvorlah}
{\sc Z.~Dvo{\v{r}}{\'a}k and A.~Lahiri}, {\em Approximation schemes for bounded
  distance problems on fractionally treewidth-fragile graphs}, arXiv,
  2105.01780 (2021).

\bibitem{dvorak2016strongly}
{\sc Z.~Dvo{\v{r}}{\'a}k and S.~Norin}, {\em Strongly sublinear separators and
  polynomial expansion}, SIAM Journal on Discrete Mathematics, 30 (2016),
  pp.~1095--1101.

\bibitem{kerdom}
{\sc C.~Einarson and F.~Reidl}, {\em {A General Kernelization Technique for
  Domination and Independence Problems in Sparse Classes}}, in 15th
  International Symposium on Parameterized and Exact Computation (IPEC 2020),
  vol.~180 of Leibniz International Proceedings in Informatics (LIPIcs),
  Dagstuhl, Germany, 2020, Schloss Dagstuhl--Leibniz-Zentrum f{\"u}r
  Informatik, pp.~11:1--11:15.

\bibitem{erlebach2005polynomial}
{\sc T.~Erlebach, K.~Jansen, and E.~Seidel}, {\em Polynomial-time approximation
  schemes for geometric intersection graphs}, SIAM Journal on Computing, 34
  (2005), pp.~1302--1323.

\bibitem{feige}
{\sc U.~Feige, M.~Hajiaghayi, and J.~R. Lee}, {\em Improved approximation
  algorithms for minimum weight vertex separators}, SIAM Journal on Computing,
  38 (2008), pp.~629--657.

\bibitem{grohe2014deciding}
{\sc M.~Grohe, S.~Kreutzer, and S.~Siebertz}, {\em Deciding first-order
  properties of nowhere dense graphs}, in Proceedings of the 46th Annual ACM
  Symposium on Theory of Computing, ACM, 2014, pp.~89--98.

\bibitem{har2015approximation}
{\sc S.~Har-Peled and K.~Quanrud}, {\em Approximation algorithms for
  polynomial-expansion and low-density graphs}, in Algorithms-ESA 2015,
  Springer, 2015, pp.~717--728.

\bibitem{kreutzer2018polynomial}
{\sc S.~Kreutzer, R.~Rabinovich, and S.~Siebertz}, {\em Polynomial kernels and
  wideness properties of nowhere dense graph classes}, ACM Transactions on
  Algorithms (TALG), 15 (2018), p.~24.

\bibitem{lt79}
{\sc R.~Lipton and R.~Tarjan}, {\em A separator theorem for planar graphs},
  SIAM Journal on Applied Mathematics, 36 (1979), pp.~177--189.

\bibitem{lt80}
\leavevmode\vrule height 2pt depth -1.6pt width 23pt, {\em Applications of a
  planar separator theorem}, SIAM Journal on Computing, 9 (1980), pp.~615--627.

\bibitem{miller1997separators}
{\sc G.~L. Miller, S.-H. Teng, W.~Thurston, and S.~A. Vavasis}, {\em Separators
  for sphere-packings and nearest neighbor graphs}, Journal of the ACM (JACM),
  44 (1997), pp.~1--29.

\bibitem{grad1}
{\sc J.~Ne{\v{s}}et\v{r}il and P.~{Ossona de Mendez}}, {\em Grad and classes
  with bounded expansion {I}. {D}ecompositions}, European J. Combin., 29
  (2008), pp.~760--776.

\bibitem{grad2}
\leavevmode\vrule height 2pt depth -1.6pt width 23pt, {\em Grad and classes
  with bounded expansion {II}. {A}lgorithmic aspects}, European J. Combin., 29
  (2008), pp.~777--791.

\bibitem{npom-nd1}
\leavevmode\vrule height 2pt depth -1.6pt width 23pt, {\em First order
  properties on nowhere dense structures}, J. Symbolic Logic, 75 (2010),
  pp.~868--887.

\bibitem{nesbook}
\leavevmode\vrule height 2pt depth -1.6pt width 23pt, {\em Sparsity (Graphs,
  Structures, and Algorithms)}, vol.~28 of Algorithms and Combinatorics,
  Springer, 2012.

\bibitem{osmenwood}
{\sc J.~Ne\v{s}et\v{r}il, P.~{Ossona de Mendez}, and D.~Wood}, {\em
  Characterisations and examples of graph classes with bounded expansion}, Eur.
  J. Comb., 33 (2012), pp.~350--373.

\bibitem{neicom}
{\sc F.~Reidl, F.~S. Villaamil, and K.~Stavropoulos}, {\em Characterising
  bounded expansion by neighbourhood complexity}, European Journal of
  Combinatorics, 75 (2019), pp.~152--168.

\bibitem{colnonap}
{\sc D.~Zuckerman}, {\em Linear degree extractors and the inapproximability of
  {M}ax {C}lique and {C}hromatic {N}umber}, Theory of Computing, 3 (2007),
  pp.~103--128.

\end{thebibliography}

\end{document}